\documentclass[a4paper, 12pt]{article}

\usepackage{amssymb}
\usepackage{amsmath}
\usepackage{amsthm}

\newcommand{\be}{\begin{equation}}
\newcommand{\ee}{\end{equation}}

\newcommand{\bC}{\mathbb{C}}
\newcommand{\bN}{\mathbb{N}}
\newcommand{\bR}{\mathbb{R}}
\newcommand{\bZ}{\mathbb{Z}}

\newcommand{\cA}{\mathcal{A}}
\newcommand{\cB}{\mathcal{B}}

\newcommand{\cF}{\mathcal{F}}
\newcommand{\cH}{\mathcal{H}}

\newcommand{\cM}{\mathcal{M}}
\newcommand{\cO}{\mathcal{O}}
\newcommand{\cP}{\mathcal{P}}
\newcommand{\cQ}{\mathcal{Q}}
\newcommand{\cR}{\mathcal{R}}

\newcommand{\cV}{\mathcal{V}}

\newcommand{\mfa}{\mathfrak{a}}
\newcommand{\mfc}{\mathfrak{c}}
\newcommand{\mfg}{\mathfrak{g}}
\newcommand{\mfgl}{\mathfrak{gl}}
\newcommand{\mfk}{\mathfrak{k}}
\newcommand{\mfm}{\mathfrak{m}}
\newcommand{\mfp}{\mathfrak{p}}
\newcommand{\mfL}{\mathfrak{L}}
\newcommand{\mfu}{\mathfrak{u}}
\newcommand{\mfX}{\mathfrak{X}}

\newcommand{\bsalpha}{\boldsymbol{\alpha}}
\newcommand{\bsbeta}{\boldsymbol{\beta}}
\newcommand{\bsa}{\boldsymbol{a}}
\newcommand{\bsb}{\boldsymbol{b}}
\newcommand{\bsc}{\boldsymbol{c}}
\newcommand{\bsd}{\boldsymbol{d}}
\newcommand{\bsu}{\boldsymbol{u}}
\newcommand{\bsone}{\boldsymbol{1}}

\newcommand{\bsLambda}{\boldsymbol{\Lambda}}

\newcommand{\bsA}{\boldsymbol{A}}
\newcommand{\bsC}{\boldsymbol{C}}

\newcommand{\bsX}{\boldsymbol{X}}

\newcommand{\tcA}{\tilde{\cA}}
\newcommand{\tbsa}{\tilde{\bsa}}
\newcommand{\tbsb}{\tilde{\bsb}}
\newcommand{\tbsc}{\tilde{\bsc}}
\newcommand{\tbsd}{\tilde{\bsd}}
\newcommand{\tbsu}{\tilde{\bsu}}

\newcommand{\hcA}{\hat{\cA}}
\newcommand{\hcB}{\hat{\cB}}
\newcommand{\hbsa}{\hat{\bsa}}
\newcommand{\hbsb}{\hat{\bsb}}
\newcommand{\hbsc}{\hat{\bsc}}
\newcommand{\hbsd}{\hat{\bsd}}
\newcommand{\hbsu}{\hat{\bsu}}

\newcommand{\ccA}{\check{\cA}}

\newcommand{\ri}{\mathrm{i}}
\newcommand{\rr}{\mathrm{r}}
\newcommand{\dd}{\mathrm{d}}
\newcommand{\diag}{\mathrm{diag}}
\newcommand{\reg}{\mathrm{reg}}
\newcommand{\tr}{\mathrm{tr}}

\newcommand{\ext}{\mathrm{ext}}
\newcommand{\red}{\mathrm{red}}
\newcommand{\Id}{\mathrm{Id}}
\newcommand{\ad}{\mathrm{ad}}
\newcommand{\wad}{\widetilde{\mathrm{ad}}}
\newcommand{\Ad}{\mathrm{Ad}}
\newcommand{\ran}{\mathrm{ran}}
\newcommand{\Real}{\mathrm{Re}}
\newcommand{\Imag}{\mathrm{Im}}

\newcommand{\eps}{\epsilon}
\newcommand{\veps}{\varepsilon}

\newcommand{\half}{\frac{1}{2}}

\newcommand{\sq}{/\!\!/}

\newcommand{\acts}{\cdot}

\newcommand{\PD}[2]{\frac{\partial #1}{\partial #2}}

\theoremstyle{plain}
\newtheorem{THEOREM}{Theorem}
\newtheorem{LEMMA}[THEOREM]{Lemma}
\newtheorem{PROPOSITION}[THEOREM]{Proposition}

\begin{document}

\begin{center}
    \Large{\textbf{On the classical $r$-matrix structure of the rational 
    $BC_n$ Ruijsenaars--Schneider--van Diejen system}}
\end{center}

\bigskip
\begin{center}
    B.G.~Pusztai \\
    Bolyai Institute, University of Szeged, \\
    Aradi v\'ertan\'uk tere 1, H-6720 Szeged, Hungary, \emph{and} \\
    MTA Lend\"ulet Holographic QFT Group, Wigner RCP, \\
    H-1525 Budapest 114, P.O.B. 49, Hungary \\
    e-mail: \texttt{gpusztai@math.u-szeged.hu}
\end{center}

\bigskip
\begin{abstract}
In this paper, we construct a quadratic $r$-matrix structure for the classical
rational $BC_n$ Ruijsenaars--Schneider--van Diejen system with the maximal
number of three independent coupling parameters. As a byproduct, we provide 
a Lax representation of the dynamics as well.



\end{abstract}
\newpage

\section{Introduction}
\label{SECTION_Introduction}
\setcounter{equation}{0}
The Ruijsenaars--Schneider--van Diejen (RSvD) models \cite{Ruij_Schneider,
van_Diejen_TMP1994} are among the most intensively studied integrable many 
particle systems, having numerous relationships with different branches of 
theoretical physics and pure mathematics. They had found applications first 
in the theory of the soliton equations \cite{Ruij_Schneider, Ruij_CMP1988, 
Babelon_Bernard, Ruij_RIMS_2, Ruij_RIMS_3}, but soon they appeared in the 
Yang--Mills and the Seiberg--Witten theories as well (see e.g. 
\cite{Gorsky_Nekrasov_1995, Nekrasov_1999, Braden_et_al_NPB1999, 
Nekrasov_Shatasvili_2009, Gadde_et_al_JHEP2014}). Besides these well-known 
links, the RSvD systems and their non-relativistic limits, the 
Calogero--Moser--Sutherland (CMS) systems \cite{Calogero, Sutherland, 
Moser_1975}, have appeared in the context of random matrix theory, too. 
Making use of the action-angle duality between the different variants of 
the CMS and the RSvD systems, new classes of random matrix ensembles emerged 
in the literature \cite{Bogomolny_et_al_PRL2009, 
Bogomolny_et_al_Nonlinearity2011, Fyodorov_Giraud_2015}, exhibiting 
spectacular statistical properties. Under the name of classical/quantum 
duality, it has also been observed that the Lax matrices of the CMS and 
the RSvD models encode the spectra of certain quantum spin chains, thereby 
the purely classical models provide an alternative way to analyze the 
quantum systems, without any reference to the celebrated Bethe Ansatz 
techniques (for details see e.g. \cite{Mukhin_et_al_2011, 
Alexandrov_et_al_NPB2014, Gorsky_et_al_JHEP2014, Tsuboi_et_al_JHEP2015}). 
It is also worth mentioning that in the recent papers 
\cite{Aminov_et_al_2014, Levin_et_al_JHEP2014} the authors have constructed 
new integrable tops, closely related to the CMS and the RSvD particle 
systems. Besides the Lax representation of the dynamics, in their studies 
the associated $r$-matrix structures also turn out to be indispensable. 

The characteristic feature the above exciting new developments all share in 
common is the prominent role played by the Lax matrices of the CMS and the 
RSvD models. However, all these investigations are based on the translational 
invariant models associated with the $A_n$ root system, exclusively. Apart 
from the technical difficulties, the probable explanation of this state of 
affair is the very limited knowledge about the Lax representation of the RSvD 
models in association with the non-$A_n$-type root systems. Of course, one 
can easily construct Lax representations for both the $C_n$-type and the 
$BC_n$-type RSvD models by the $\bZ_2$-folding of the $A_{2 n - 1}$ and the 
$A_{2 n}$ root systems, respectively \cite{Chen_et_al_JMP2000}. However, this 
trivial approach is only of very limited use, since the resulting models 
contain only a single coupling parameter. Nevertheless, working in a symplectic 
reduction framework, in our papers \cite{Pusztai_NPB2011, Pusztai_NPB2012} 
we succeeded in constructing Lax matrices for the rational $C_n$ and the
rational $BC_n$ RSvD systems with the maximal number of independent coupling
constants. Motivated by the plethora of potential applications outlined above, 
in this paper we work out the underlying classical $r$-matrix structures and 
also provide a Lax representation of the dynamics for the rational $BC_n$ RSvD 
model with three independent coupling parameters. 

Let us recall that the configuration space of the rational $BC_n$ RSvD system 
is the open subset
\be
    \mfc 
    = \{ \lambda = (\lambda_1, \ldots, \lambda_n) \in \bR^n
        \, | \,
        \lambda_1 > \ldots > \lambda_n > 0 \}
    \subseteq \bR^n,
\label{mfc}
\ee
that can be seen as an appropriate model for the standard open Weyl chamber
of type $BC_n$. The cotangent bundle $T^* \mfc$ is trivial, whence the 
phase space of the RSvD system can be identified with the product manifold
\be
    \cP^R 
    = \mfc \times \bR^n
    = \{ (\lambda, \theta) 
        \, | \, 
        \lambda \in \mfc \, , \theta \in \bR^n \},
\label{cP_R}
\ee
that we endow with the symplectic form
\be
    \omega^R = 2 \sum_{c = 1}^n \dd \theta_c \wedge \dd \lambda_c.
\label{omega_R}
\ee
We mention in passing that the unusual numerical factor in $\omega^R$ is 
inserted purely for consistency with our earlier works 
\cite{Pusztai_NPB2011, Pusztai_NPB2012}. As for the dynamics, it is 
governed by the Hamiltonian
\be
\begin{split}
    H^R 
    = & \sum_{c = 1}^{n} \cosh(2 \theta_c)
        \left( 
            1 + \frac{\nu^2}{\lambda_c^2} 
        \right)^\frac{1}{2}
        \left( 
            1 + \frac{\kappa^2}{\lambda_c^2} 
        \right)^\frac{1}{2}
        \prod_{\substack{d = 1 \\ (d \neq c)}}^{n}
        \left( 
            1 + \frac{4 \mu^2}{(\lambda_c - \lambda_d)^2} 
        \right)^\frac{1}{2}
        \left( 
            1 + \frac{4 \mu^2}{(\lambda_c + \lambda_d)^2} 
        \right)^\frac{1}{2}
    \\
    & + \frac{\nu \kappa}{4 \mu^2} 
        \prod_{c = 1}^n 
            \left(
                1 + \frac{4 \mu^2}{\lambda_c^2}
            \right)
        - \frac{\nu \kappa}{4 \mu^2},
\label{H_R}
\end{split}
\ee
where $\mu$, $\nu$ and $\kappa$ are arbitrary real parameters satisfying 
$\mu < 0 < \nu$. Also, on these so-called coupling constants in this paper 
we impose the condition $\nu \kappa \geq 0$. As can be seen in 
\cite{Pusztai_NPB2013}, this additional requirement ensures that the 
particle system possesses only scattering trajectories. 

Having defined the models of our interest, now we wish to outline the 
content of the rest of the paper. To keep our present work essentially 
self-contained, in Section \ref{SECTION_Preliminaries} we briefly skim 
through the necessary Lie theoretic machinery and the symplectic reduction
background, that provide the building blocks of the latter developments. 
Also, this section allows us to fix the notations. Starting with Section 
\ref{SECTION_C_r_matrix} we present our new results. Section 
\ref{SECTION_C_r_matrix} is the longest and the most technical part 
of our paper, in which we study of the $r$-matrix structure of the rational 
$C_n$ RSvD model corresponding to the special choice $\kappa = 0$. Sticking 
to the Marsden--Weinstein reduction approach, in Subsection 
\ref{SUBSECTION_local_extension} we construct local extensions for the Lax 
matrix of the rational $C_n$ RSvD model. Making use of these local sections, 
in Subsection \ref{SUBSECTION_computing_r} a series of short Propositions 
and Lemmas allows us to construct a classical $r$-matrix structure for the 
$C_n$-type model. In this respect our main result is Theorem 
\ref{THEOREM_C_n_r_matrix}, in which we formulate the $r$-matrix structure 
in a convenient quadratic form. The resulting quadratic $r$-matrices turn out 
to be fully dynamical, depending on all variables of the phase space $\cP^R$. 
Subsequently, by switching to a purely algebraic approach, in Section 
\ref{SECTION_BC_r_matrix} we generalize Theorem \ref{THEOREM_C_n_r_matrix} 
to the rational $BC_n$ RSvD system with three independent coupling constants. 
The quadratic $r$-matrix structure of the $BC_n$-type system is summarized 
in Theorem \ref{THEOREM_BC_n_quadratic_PB}. To make this important result 
more transparent, in Theorem \ref{THEOREM_BC_n_transformed_r_matrices} we 
describe the $r$-matrix structure in a more convenient choice of gauge. 
In this gauge we also provide a Lax representation of the dynamics, as 
formulated in Theorem \ref{THEOREM_Lax_pair}. Finally, in Section 
\ref{SECTION_Discussion} we offer a short discussion on our results and also 
point out some open problems related to the RSvD systems.

\section{Preliminaries}
\label{SECTION_Preliminaries}
\setcounter{equation}{0}
In this section we overview those Lie theoretic notions and results that 
underlie the geometric construction of the classical $r$-matrix structure 
for the rational $C_n$ RSvD system. Our approach is based on the symplectic 
reduction derivation of the RSvD models, that we also briefly outline. In 
Subsection \ref{SUBSECTION_Lie_theoretic_background} we closely follow the 
conventions of the standard reference \cite{Knapp}, whereas in Subsection 
\ref{SUBSECTION_C_RSvD} we employ the notations introduced in our earlier 
work \cite{Pusztai_NPB2011} on the RSvD systems.

\subsection{Lie theoretic background}
\label{SUBSECTION_Lie_theoretic_background}
Take a positive integer $n \in \bN$ and keep it fixed. Let $N = 2 n$ and 
introduce the sets
\be 
    \bN_n = \{ 1, \ldots, n \}
    \quad \text{and} \quad
    \bN_N = \{ 1, \ldots, N \}.
\label{bN_n} 
\ee
With the aid of the $N \times N$ matrix
\be
    \bsC 
    = \begin{bmatrix}
    	0_n & \bsone_n \\
    	\bsone_n & 0_n
    \end{bmatrix}
\label{bsC} 
\ee
we define the non-compact real reductive matrix Lie group
\be
	G = U(n, n) = \{ y \in GL(N, \bC) \, | \, y^* \bsC y = \bsC \},
\label{G}
\ee
that we equip with the Cartan involution
\be
    \varTheta \colon G \rightarrow G,
    \quad
    y \mapsto (y^{-1})^*.
\label{varTheta}
\ee
Its fixed-point set
\be
    K = \{ y \in G \, | \, \varTheta(y) = y \}
\label{K}
\ee
is a maximal compact subgroup of $G$, having the identification
$K \cong U(n) \times U(n)$.

On the Lie algebra
\be
	\mfg 
	= \mfu(u, n) 
	= \{ Y \in \mfgl(N, \bC) \, | \, Y^* \bsC + \bsC Y = 0 \}
\label{mfg}
\ee
the corresponding involution
\be
    \vartheta
    \colon \mfg \rightarrow \mfg,
    \quad
    Y \mapsto - Y^*
\label{vartheta}
\ee
naturally induces the Cartan decomposition
\be
    \mfg = \mfk \oplus \mfp
\label{gradation}
\ee
with the Lie subalgebra and the complementary subspace
\be
    \mfk = \ker(\vartheta - \Id_\mfg) 
    \quad \text{and} \quad
    \mfp = \ker(\vartheta + \Id_\mfg),
\label{mfkp}
\ee
respectively. That is, each element $Y \in \mfg$ can be decomposed as
\be
    Y = Y_+ + Y_-
\label{Y+-}
\ee
with unique components $Y_+ \in \mfk$ and $Y_- \in \mfp$. Notice that the 
$\bZ_2$-gradation (\ref{gradation}) of $\mfg$ is actually orthogonal with 
respect to the non-degenerate $\Ad$-invariant symmetric bilinear form
\be
    \langle \, , \rangle \colon 
    \mfg \times \mfg \rightarrow \bR,
    \quad
    (Y_1, Y_2) \mapsto \tr(Y_1 Y_2).
\label{bilinear_form}
\ee

To make our presentation simpler, for all $k, l \in \bN_N$ we introduce the
standard elementary matrix $e_{k, l} \in \mfgl(N, \bC)$ with entries
\be
    (e_{k, l})_{k', l'} = \delta_{k, k'} \delta_{l, l'}
    \qquad
    (k', l' \in \bN_N).
\label{e_kl}
\ee
Also, with each $\lambda = (\lambda_1, \ldots, \lambda_n) \in \bR^n$ we 
associate the $N \times N$ diagonal matrix
\be
    \bsLambda(\lambda) 
    = \diag(\lambda_1, \ldots, \lambda_n, -\lambda_1, \ldots, -\lambda_n)
    \in \mfp.
\label{bsLambda}
\ee
The set of diagonal matrices
\be
    \mfa = \{ \bsLambda(\lambda) \, | \, \lambda \in \bR^n \} 
\label{mfa}
\ee
forms a maximal Abelian subspace in $\mfp$. Note that in $\mfa$ the family of 
matrices
\be
    D_c^- = \frac{1}{\sqrt{2}} (e_{c, c} - e_{n + c, n + c})
    \qquad
    (c \in \bN_n)
\label{D_-_c}
\ee
forms an orthonormal basis, i.e. 
$\langle D^-_c, D^-_d \rangle = \delta_{c, d}$ for all $c, d \in \bN_n$.

The centralizer of the Lie algebra $\mfa$ inside $K$ is the Abelian Lie group
\be
    M = Z_K(\mfa)
    = \{ \diag(e^{\ri \chi_1}, \ldots, e^{\ri \chi_n}, 
                e^{\ri \chi_1}, \ldots, e^{\ri \chi_n}) 
        \, | \,
        \chi_1, \ldots, \chi_n \in \bR \}
\label{M}
\ee
with Lie algebra
\be
    \mfm = \{ \diag(\ri \chi_1, \ldots, \ri \chi_n, 
                    \ri \chi_1, \ldots, \ri \chi_n) 
            \, | \,
            \chi_1, \ldots, \chi_n \in \bR \}.
\label{mfm}
\ee
In this Abelian Lie algebra the set of matrices
\be
    D_c^+ = \frac{\ri}{\sqrt{2}} (e_{c, c} + e_{n + c, n + c})
    \qquad
    (c \in \bN_n)
\label{D_+_c}
\ee
forms a basis obeying the orthogonality relations
$\langle D^+_c, D^+_d \rangle = - \delta_{c, d}$ $(c, d \in \bN_n)$.

Let $\mfm^\perp$ and $\mfa^\perp$ denote the sets of the off-diagonal elements 
of $\mfk$ and $\mfp$, respectively. With these subspaces can write the refined 
orthogonal decomposition
\be
    \mfg = \mfm \oplus \mfm^\perp \oplus \mfa \oplus \mfa^\perp.
\label{refined_decomposition}
\ee
In other words, each element $Y \in \mfg$ can be uniquely decomposed as
\be
    Y = Y_\mfm + Y_{\mfm^\perp} + Y_\mfa + Y_{\mfa^\perp},
\label{Y_decomp}
\ee
where each component belongs to the subspace indicated by the subscript. In 
order to provide convenient bases in the subspaces $\mfm^\perp$ and 
$\mfa^\perp$, for each $c \in \bN_n$ we introduce the linear functional
\be
	\veps_c \colon \bR^n \rightarrow \bR,
	\quad
	\lambda = (\lambda_1, \ldots, \lambda_n) \mapsto \lambda_c.
\label{varepsilon_c}
\ee
Let us observe that the set of functionals
\be
	\cR_+ 
	= \{ \veps_a \pm \veps_b \, | \, 1 \leq a < b \leq n \}
	\cup
	\{2 \veps_c \, | \, c \in \bN_n \}	
\label{cR_+}
\ee
can be seen as a realization of a set of \emph{positive} roots of type $C_n$. 
Now, associated with the positive root $2 \veps_c$ $(c \in \bN_n)$, we define 
the matrices
\be
	X_{2 \veps_c}^{\pm, \ri} 
	= -\frac{\ri}{\sqrt{2}} (e_{c, n + c} \pm e_{n + c, c}).
\label{X_+_i_2}
\ee
In association with the other positive roots, for all $1 \leq a < b \leq n$ 
we define the following matrices with purely real entries:
\be
\begin{split}
	X_{\veps_a - \veps_b}^{\pm, \rr} &
	= \half (e_{a, b} \mp e_{b, a} 
				\pm e_{n + a, n + b} - e_{n + b, n + a}), \\
	\quad
	X_{\veps_a + \veps_b}^{\pm, \rr} &
	= - \half (e_{a, n + b} - e_{b, n + a} 
				\pm e_{n + a, b} \mp e_{n + b, a}), 
\end{split}
\label{X_pm_rrri_pm}
\ee
together with the following ones with purely imaginary entries:
\be
\begin{split}
	X_{\veps_a - \veps_b}^{\pm, \ri} &
	= \frac{\ri}{2} (e_{a, b} \pm e_{b, a} 
					\pm e_{n + a, n + b} + e_{n + b, n + a}), \\
	\quad
	X_{\veps_a + \veps_b}^{\pm, \ri} &
	= - \frac{\ri}{2} (e_{a, n + b} + e_{b, n + a} 
					\pm e_{n + a, b} \pm e_{n + b, a}). 
\end{split}
\ee
The point is that the set of vectors $\{ X_\alpha^{+, \eps} \}$ forms a basis 
in the subspace $\mfm^\perp$, whereas the family $\{ X_\alpha^{-, \eps} \}$ 
provides a basis in $\mfa^\perp$. Moreover, they obey the orthogonality 
relations
\be
	\langle X^{+, \eps}_\alpha, X^{+, \eps'}_{\alpha'} \rangle
	= - \delta_{\alpha, \alpha'} \delta_{\eps, \eps'}
	\quad \text{and} \quad 
	\langle X^{-, \eps}_\alpha, X^{-, \eps'}_{\alpha'} \rangle
	= \delta_{\alpha, \alpha'} \delta_{\eps, \eps'}.
\label{X_scalar_product}
\ee
Note that the family of vectors
\be
    \{ v_I \} \equiv \{ D^\pm_c \} \cup \{ X^{\pm, \eps}_\alpha \}
\label{v_I}
\ee
forms a basis in the real Lie algebra $\mfu(n, n)$. We mention in passing 
that it is a basis in the complexification 
$\mfgl(N, \bC) \cong \mfu(n, n)^\bC$, too.

Next we turn to the linear operator
\be
    \ad_{\bsLambda(\lambda)} \colon \mfg \rightarrow \mfg,
    \quad
    Y \mapsto [\bsLambda(\lambda), Y],
\label{ad_bsLambda}
\ee
defined for each $\lambda \in \bR^n$. The real convenience of the basis 
(\ref{v_I}) stems from the commutation relations
\be
    \ad_{\bsLambda(\lambda)}(D^\pm_c) = 0
    \quad \text{and} \quad
    \ad_{\bsLambda(\lambda)}(X^{\pm, \eps}_\alpha)
    = \alpha(\lambda) X^{\mp, \eps}_\alpha,
\label{commut_rel}
\ee
where $c \in \bN_n$, $\alpha \in \cR_+$ and $\eps \in \{ \rr, \ri \}$. Notice 
that the subspace $\mfm^\perp \oplus \mfa^\perp$ is invariant under the linear 
operator $\ad_{\bsLambda(\lambda)}$, whence the restriction
\be
    \wad_{\bsLambda(\lambda)} 
    = \ad_{\bsLambda(\lambda)} |_{\mfm^\perp \oplus \mfa^\perp}
    \in \mfgl(\mfm^\perp \oplus \mfa^\perp)
\label{wad}
\ee
is well-defined for all $\lambda \in \bR^n$, with spectrum
\be
    \text{Spec}(\wad_{\bsLambda(\lambda)})
    = \{ \pm \alpha(\lambda) \, | \, \alpha \in \cR_+ \}.
\label{tilde_ad_spectrum}
\ee
The regular part of $\mfa$ is defined by the subset
\be
    \mfa_\reg 
    = \{ \bsLambda(\lambda) 
        \, | \, 
        \lambda \in \bR^n 
        \text{ and } 
        \wad_{\bsLambda(\lambda)} \text{ is invertible} \},
\label{mfa_reg}
\ee
in which the standard Weyl chamber 
$\{ \bsLambda(\lambda) \, | \, \lambda \in \mfc \}$
is an appropriate connected component. Note that this Weyl chamber can be
naturally identified with the configuration space $\mfc$ (\ref{mfc}) of the 
rational $BC_n$ RSvD system.

Having set up the algebraic stage, now we turn to some geometric results that 
are specific to the symplectic reduction derivation of the rational RSvD 
models. First, recall that the regular part of $\mfp$ (\ref{mfkp}) defined by
\be
    \mfp_\reg 
    = \{ k \bsLambda(\lambda) k^{-1} 
        \, | \,
        \lambda \in \mfc \text{ and } k \in K \}
\label{mfp_reg}
\ee
is a dense and open subset of $\mfp$. It is an important fact that with the
smooth free right $M$-action
\be
    M \times (\mfc \times K) \ni 
        (m, (\lambda, k)) \mapsto (\lambda, k m)
    \in \mfc \times K 
\label{M_action}
\ee
the map
\be
    \pi \colon \mfc \times K \twoheadrightarrow \mfp_\reg,
    \quad
    (\lambda, k) \mapsto k \bsLambda(\lambda) k^{-1}
\label{pi}
\ee
is a smooth principal $M$-bundle, providing the identification
\be
    \mfp_\reg \cong (\mfc \times K) / M \cong \mfc \times (K / M).
\label{mfp_reg_identification}
\ee
In the geometric construction of the dynamical $r$-matrix for the rational 
$C_n$ RSvD model we shall utilize certain local sections of $\pi$ with
the characteristic properties below.

\begin{PROPOSITION}
\label{PROPOSITION_e_and_sigma}
Take an arbitrary point $\lambda^{(0)} \in \mfc$ and let 
$\bsLambda^{(0)} = \bsLambda(\lambda^{(0)})$. Then there is a smooth
local section
\be
    \check{\mfp}_\reg \ni 
        Y \mapsto (e(Y), \sigma(Y))
    \in \mfc \times K
\label{e_and_sigma}
\ee
of $\pi$ (\ref{pi}), defined on some open subset 
$\check{\mfp}_\reg \subseteq \mfp_\reg$, such that
\be
    \bsLambda^{(0)} \in \check{\mfp}_\reg,
    \quad
    (e(\bsLambda^{(0)}), \sigma(\bsLambda^{(0)})) = (\lambda^{(0)}, \bsone),
    \quad
    \ran(\sigma_{* \bsLambda^{(0)}}) \subseteq \mfm^\perp.
\label{e_and_sigma_cond}
\ee
Moreover, under these conditions, at the point $\bsLambda^{(0)}$ the action 
of the derivatives of $e$ and $\sigma$ on the tangent vector 
$\delta Y \in \mfp \cong T_{\bsLambda^{(0)}} \check{\mfp}_\reg$ takes the form
\begin{align}
    & e_{* \bsLambda^{(0)}}(\delta Y)
    = (\delta Y_{1, 1}, \ldots, \delta Y_{n, n}) 
        \in \bR^n \cong T_{\lambda^{(0)}} \mfc, 
    \label{e_derivative} \\
    & \sigma_{* \bsLambda^{(0)}}(\delta Y)
    = - (\wad_{\bsLambda^{(0)}})^{-1} ((\delta Y)_{\mfa^\perp})
        \in \mfm^\perp \subseteq \mfk \cong T_{\bsone} K.
    \label{sigma_derivative}
\end{align}
\end{PROPOSITION}

\begin{proof}
Notice that the point $(\lambda^{(0)}, \bsone) \in \mfc \times K$ projects 
onto $\bsLambda^{(0)}$, that is, 
\be
    \pi(\lambda^{(0)}, \bsone) = \bsLambda^{(0)}. 
\label{P1_projection}
\ee    
Differentiating $\pi$ (\ref{pi}) at $(\lambda^{(0)}, \bsone)$, let us observe 
that for each tangent vector
\be
    \delta \lambda \oplus \delta k \in 
        \bR^n \oplus \mfk 
        \cong
        T_{\lambda^{(0)}} \mfc \oplus T_{\bsone} K
        \cong
        T_{(\lambda^{(0)}, \bsone)} (\mfc \times K)
\label{P1_tangent_vector}
\ee
we can write
\be
    \pi_{* (\lambda^{(0)}, \bsone)} (\delta \lambda \oplus \delta k)
    = \left \{ 
        \frac{\dd}{\dd t} 
        \pi(\lambda^{(0)} + t \delta \lambda, e^{t \delta k})
        \right \}_{t = 0}
    = \bsLambda(\delta \lambda) - [ \bsLambda^{(0)}, \delta k ].
\label{P1_pi_der}
\ee
Utilizing the linear operator (\ref{wad}), it is clear that
\be
    \pi_{* (\lambda^{(0)}, \bsone)} (\delta \lambda \oplus \delta k)
    = \bsLambda(\delta \lambda) 
        - \wad_{\bsLambda^{(0)}}((\delta k)_{\mfm^\perp}),
\label{P1_pi_der_OK}
\ee
from where we conclude that
\be
    \ker(\pi_{* (\lambda^{(0)}, \bsone)}) = \{ 0 \} \oplus \mfm.
\label{P1_pi_der_kernel}
\ee
Since the subspace $\bR^n \oplus \mfm^\perp$ is a complementary subspace of 
$\ker(\pi_{* (\lambda^{(0)}, \bsone)})$ in the tangent space
$T_{(\lambda^{(0)}, \bsone)} (\mfc \times K)$, it is evident that there 
exists a local section $(e, \sigma)$ (\ref{e_and_sigma}) satisfying the 
conditions imposed in (\ref{e_and_sigma_cond}). Moreover, by differentiating 
the equation
\be
    \pi \circ (e, \sigma) = \Id_{\check{\mfp}_\reg}
\label{P1_section_rel}
\ee
at the point $\bsLambda^{(0)}$, the relationship (\ref{P1_pi_der_OK}) entails 
that for all $\delta Y \in \mfp \cong T_{\bsLambda^{(0)}} \check{\mfp}_\reg$ 
we can write
\be
    \bsLambda(e_{* \bsLambda^{(0)}} (\delta Y))
        - \wad_{\bsLambda^{(0)}}(\sigma_{* \bsLambda^{(0)}}(\delta Y))
    = \delta Y.
\label{P1_section_der}
\ee
Projecting this equation onto the subspaces $\mfa$ and $\mfa^\perp$, 
respectively, the formulae for the derivatives displayed in 
(\ref{e_derivative}) and (\ref{sigma_derivative}) follow at once.
\end{proof}

To proceed further, we introduce the set of complex column vectors
\be
    S = \{ V \in \bC^N \, | \, \bsC V + V = 0 \text{ and } V^* V = N \},
\label{S}
\ee
that can be naturally identified with a sphere of real dimension $2 n - 1$.
At each point $V \in S$ the tangent space to $S$ can be identified with 
the real subspace of the complex column vectors
\be
    T_V S 
    = \{ \delta V \in \bC^N 
        \, | \, 
        \bsC \delta V + \delta V = 0 
        \text{ and }
        (\delta V)^* V + V^* \delta V = 0 
    \}, 
\label{TS}
\ee
that we endow with the inner product
\be
    \langle \delta V, \delta v \rangle_{T_V S} 
    = \Real((\delta V)^* \delta v)
    \qquad
    (\delta V, \delta v \in T_V S).
\label{TS_inner_prod}
\ee
Next, we introduce the distinguished column vector $E \in S$ with components
\be
    E_a = 1 \quad \text{and} \quad E_{n + a} = -1 \qquad (a \in \bN_n).
\label{E}
\ee
Also, with each vector $V \in S$ we associate the $N \times N$ matrix
\be
    \xi(V) = \ri \mu (V V^* - \bsone) + \ri (\mu - \nu) \bsC \in \mfk.
\label{xi_def}
\ee
Since the $K$-action on $S$ defined by the smooth map
\be
    K \times S \ni (k, V) \mapsto k V \in S
\label{K_action_on_S}
\ee
is transitive, and since $k \xi(V) k^{-1} = \xi(k V)$ for all $k \in K$ and 
$V \in S$, it is clear that the adjoint orbit of $K$ passing through the 
element $\xi(E) \in \mfk$ has the form
\be
    \cO 
    = \{ k \xi(E) k^{-1} \, | \, k \in K \} 
    = \{ \xi(V) \, | \, V \in S \}.
\label{cO}
\ee
As is known, the orbit $\cO$ can be seen as an embedded submanifold of $\mfk$, 
and for its tangent spaces we have the identifications
\be
    T_\rho \cO = \{ [X, \rho] \, | \, X \in \mfk \} \subseteq \mfk
    \qquad
    (\rho \in \cO).
\label{cO_tangent_space}
\ee
In our earlier papers \cite{Feher_Pusztai_NPB2006, Feher_Pusztai_LMP2007, 
Pusztai_NPB2011, Pusztai_NPB2012} we have seen many times that this 
non-trivial minimal adjoint orbit plays a distinguished role in the symplectic 
reduction derivation of both the CMS and the RSvD systems. In this paper,
throughout the construction of a dynamical $r$-matrix for the rational $C_n$ 
RSvD system, we will also exploit that with the free $U(1)$-action
\be
    U(1) \times S \ni (e^{\ri \psi}, V) \mapsto e^{\ri \psi} V \in S
\label{U(1)_action_on_S}
\ee
the map
\be
    \xi \colon S \twoheadrightarrow \cO,
    \quad
    V \mapsto \xi(V) 
\label{xi}
\ee
is a smooth principal $U(1)$-bundle, providing the identification
$\cO \cong S / U(1)$. Recalling (\ref{xi_def}), it is clear that the 
derivative of $\xi$ takes the form
\be
    \xi_{* V} (\delta V) 
    = \ri \mu ((\delta V) V^* + V (\delta V)^*) \in T_{\xi(V)} \cO
    \qquad
    (V \in S, \: \delta V \in T_V S),
\label{xi_der}
\ee
whence it follows that 
\be
    \ker(\xi_{* V}) = \bR \ri V
    \quad \text{and} \quad
    (\ker(\xi_{* V}))^\perp 
    = \{ \delta v \in T_V S 
        \, | \,
        (\delta v)^* V = V^* \delta v
    \}.
\label{xi_der_kernel}
\ee
Let us also note that for all $X \in \mfk$ and $V \in S$ we have
$X V \in T_V S$ and
\be
    \xi_{* V} (X V) = [X, \xi(V)] \in T_{\xi(V)} \cO.
\label{xi_der_comm_rel}
\ee
The last two equations entail that for each $\delta V \in T_V S$ one can find 
a Lie algebra element $X \in \mfk$ and a real number $t \in \bR$ such that
\be
    \delta V = X V + t \ri V.
\label{delta_V_and_X_and_t}
\ee
Having determined the derivative of $\xi$, now we shall work out the 
derivatives of certain local sections, that find applications it the latter 
developments.

\begin{PROPOSITION}
\label{PROPOSITION_W}
Let $\cV^{(0)} \in S$ be an arbitrary point and define
$\rho^{(0)} = \xi(\cV^{(0)}) \in \cO$. Take a smooth local section
\be
    W \colon \check{\cO} \rightarrow S,
    \quad
    \rho \mapsto W(\rho)
\label{W}
\ee
of $\xi$ (\ref{xi}), defined on some open subset 
$\check{\cO} \subseteq \cO$, satisfying the conditions
\be
    \rho^{(0)} \in \check{\cO},
    \quad
    W(\rho^{(0)}) = \cV^{(0)},
    \quad
    \ran(W_{* \rho^{(0)}}) \subseteq (\ker(\xi_{* \cV^{(0)}}))^\perp.
\label{W_cond}
\ee 
Then for the derivative of $W$ at the point $\rho^{(0)}$ we have
\be
    W_{* \rho^{(0)}}([X, \rho^{(0)}])
    = X \cV^{(0)} - \frac{(\cV^{(0)})^* X \cV^{(0)}}{N} \cV^{(0)}
    \qquad
    (X \in \mfk).
\label{W_der}
\ee
\end{PROPOSITION}

\begin{proof}
It is evident that there is a smooth local section $W$ of the principal 
$U(1)$-bundle $\xi$ that satisfies the conditions displayed in 
(\ref{W_cond}). Take an arbitrary tangent vector 
$[X, \rho^{(0)}] \in T_{\rho^{(0)}} \cO$ generated by some $X \in \mfk$, 
and introduce the shorthand notation
\be
    \delta W 
    = W_{* \rho^{(0)}}([X, \rho^{(0)}]) 
        \in (\ker(\xi_{* \cV^{(0)}}))^\perp.
\label{delta_W}
\ee
By taking the derivative of the relationship 
$\xi \circ W = \Id_{\check{\cO}}$ at the point $\rho^{(0)}$, we find that
\be
    \xi_{* \cV^{(0)}}(\delta W)
    = \xi_{* \cV^{(0)}} \circ W_{* \rho^{(0)}}([X, \rho^{(0)}])
    = [X, \rho^{(0)}]
    = \xi_{* \cV^{(0)}} (X \cV^{(0)}),
\label{xi_der_delta_W}
\ee
therefore $\delta W - X \cV^{(0)} \in \ker(\xi_{* \cV^{(0)}})$. However, 
due to (\ref{xi_der_kernel}) we can write that
\be
    \delta W = X \cV^{(0)} + x \ri \cV^{(0)}
\label{delta_W_and_x}
\ee
with a unique real number $x$. Its value can determined by the fact 
that the tangent vector $\delta W$ belongs to subspace 
$(\ker(\xi_{* \cV^{(0)}}))^\perp$, leading to the formula (\ref{W_der}). 
\end{proof}

\subsection{The rational $C_n$ RSvD model from symplectic reduction}
\label{SUBSECTION_C_RSvD}
Based on our earlier results, in this subsection we review the symplectic
reduction derivation of the rational $C_n$ RSvD system. The surrounding ideas 
and the proofs can be found in \cite{Pusztai_NPB2011}. An important ingredient 
of the symplectic reduction derivation of the RSvD system of our interest is 
the cotangent bundle $T^* G$ of the Lie group $G$ (\ref{G}). For convenience, 
we trivialize $T^* G$ by the left translations. Moreover, by identifying the 
dual space $\mfg^*$ with the Lie algebra $\mfg$ (\ref{mfg}) via the bilinear 
form (\ref{bilinear_form}), it is clear that the product manifold
$\cP = G \times \mfg$ provides an appropriate model for $T^* G$. For the 
tangent spaces of the manifold $\cP$ we have the natural identifications
\be
    T_{(y, Y)} \cP 
    \cong T_y G \oplus T_Y \mfg \cong T_y G \oplus \mfg
    \qquad
    ((y, Y) \in \cP),
\label{cP_tangent_spaces}
\ee
and for the canonical symplectic form $\omega \in \Omega^2(\cP)$ 
we can write
\be
    \omega_{(y, Y)}(\Delta y \oplus \Delta Y, \delta y \oplus \delta Y)
    = \langle y^{-1} \Delta y, \delta Y \rangle
        - \langle y^{-1} \delta y, \Delta Y \rangle
        + \langle [y^{-1} \Delta y, y^{-1} \delta y], Y \rangle,
\label{omega}
\ee
where $(y, Y) \in \cP$ is an arbitrary point and 
$\Delta y \oplus \Delta Y, \delta y \oplus \delta Y \in T_y G \oplus \mfg$
are arbitrary tangent vectors. An equally important building block in the 
geometric picture underlying reduction derivation of the RSvD model is the 
adjoint orbit $\cO$ (\ref{cO}). Of course, it carries the 
Kirillov--Kostant--Souriau symplectic form $\omega^\cO \in \Omega^2(\cO)$, 
that can be written as
\be
    \omega^\cO_\rho( [X, \rho], [Z, \rho] ) 
    = \langle \rho, [X, Z] \rangle
    \qquad
    (\rho \in \cO, \: X, Z \in \mfk).
\label{omega_cO}
\ee
Making use of the bundle $\xi$ (\ref{xi}) and the equations 
(\ref{xi_der_comm_rel}) and (\ref{delta_V_and_X_and_t}), one can easily 
see that
\be
    \omega^\cO_{\xi(V)}(\xi_{* V}(\delta V), \xi_{* V}(\delta v))
    = 2 \mu \Imag((\delta V)^* \delta v)
    \qquad
    (V \in S, \: \delta V, \delta v \in T_V S).
\label{omega_cO_OK}
\ee

Now, by taking the symplectic product of the symplectic manifolds 
$(\cP, \omega)$ and $(\cO, \omega^\cO)$, we introduce the extended phase space
\be
    (\cP^\ext, \omega^\ext) = (\cP \times \cO, \omega + \omega^\cO).
\label{cP_ext}
\ee
To describe the Poisson bracket on this space, for each smooth function 
$F \in C^\infty(\cP^\ext)$, at each point $u = (y, Y, \rho) \in \cP^\ext$, 
we define the gradients
\be
    \nabla^G F(u) \in \mfg,
    \quad
    \nabla^\mfg F(u) \in \mfg,
    \quad
    \nabla^\cO F(u) \in T_\rho \cO
\label{F_gradients}
\ee
by the natural requirement
\be
    F_{* u}(\delta y \oplus \delta Y \oplus [X, \rho])
    = \langle \nabla^G F(u), y^{-1} \delta y \rangle
        + \langle \nabla^\mfg F(u), \delta Y \rangle
        + \langle \nabla^\cO F(u), X \rangle,
\label{F_gradients_def}
\ee
where $\delta y \in T_y G$, $\delta Y \in \mfg$ and $X \in \mfk$ are
arbitrary elements. Now, one can easily verify that the Poisson bracket on 
$\cP^\ext$ induced by the symplectic form $\omega^\ext$ can be cast into 
the form
\be
\begin{split}
	\{ F, H \}^\ext(u)
	& = \langle \nabla^G F(u), \nabla^\mfg H(u) \rangle
		- \langle \nabla^\mfg F(u), \nabla^G H(u) \rangle \\
    	& \quad - \langle [\nabla^\mfg F(u), \nabla^\mfg H(u)], Y \rangle
		+ \omega^\cO_\rho (\nabla^\cO F(u), \nabla^\cO H(u)),
\end{split}
\label{PB_ext}
\ee
for all $F, H \in C^\infty(\cP^\ext)$. To proceed further, let us note 
that the smooth map
\be
	\Phi^\ext \colon (K \times K) \times \cP^\ext \rightarrow \cP^\ext,
	\quad
	((k_L , k_R), (y, Y, \rho)) 
	\mapsto 
	(k_L y k_R^{-1}, k_R Y k_R^{-1}, k_L \rho k_L^{-1})  
	\label{Phi_ext}
\ee
is a symplectic left action of the product Lie group $K \times K$ on the 
extended phase space $\cP^\ext$, and it admits a $K \times K$-equivariant 
momentum map
\be
	J^\ext \colon \cP^\ext \rightarrow \mfk \oplus \mfk,
	\quad
	(y, Y, \rho)
	\mapsto
	( (y Y y^{-1})_+ + \rho ) \oplus (- Y_+).
	\label{J_ext}
\ee
As we proved in \cite{Pusztai_NPB2011}, the rational $C_n$ RSvD model
can be derived by reducing the symplectic manifold $\cP^\ext$ at the 
zero value of the momentum map $J^\ext$.

Let us recall that the standard Marsden--Weinstein reduction consists of two 
major steps. At the outset, we need control over the level set
\be
    \mfL_0 
    = (J^\ext)^{-1}(\{ 0 \}) 
    = \{ u \in \cP^\ext \, | \, J^\ext(u) = 0 \},
\label{mfL_0}
\ee
that turns out to be an embedded submanifold of $\cP^\ext$ (\ref{cP_ext}). 
However, to get a finer picture, we still need some more background material. 
First, for each $a \in \bN_n$ we define the rational function
\be
    \mfc \ni \lambda
    \mapsto
    z_a(\lambda) 
        = - \left(1 + \frac{\ri \nu}{\lambda_a} \right)
            \prod_{\substack{d = 1 \\ (d \neq a)}}^n
            \left( 1 + \frac{2 \ri \mu}{\lambda_a - \lambda_d} \right)
            \left( 1 + \frac{2 \ri \mu}{\lambda_a + \lambda_d} \right) 
    \in \bC.
\label{z}
\ee
Also, we need the vector-valued function $\cF \colon \cP^R \rightarrow \bC^N$ 
with components
\be 
    \cF_a = e^{\theta_a} \vert z_a \vert^\half
    \quad \text{and} \quad
    \cF_{n + a} = e^{-\theta_a} \overline{z}_a \vert z_a \vert^{-\half}
    \qquad
    (a \in \bN_n),
\label{cF}
\ee
that allows us to introduce the function 
$\cA \colon \cP^R \rightarrow \exp(\mfp)$ with the matrix entries
\be
\begin{split}
    & \cA_{a, b} 
    = \frac{2 \ri \mu \cF_a \overline{\cF}_b}
        {2 \ri \mu + \lambda_a - \lambda_b}, 
    \quad
    \cA_{n + a, n + b} 
    = \frac{2 \ri \mu \cF_{n + a} \overline{\cF}_{n + b}}
        {2 \ri \mu - \lambda_a + \lambda_b}, \\
    & \cA_{a, n + b} 
    = \overline{\cA}_{n + b, a} 
    = \frac{2 \ri \mu \cF_a \overline{\cF}_{n + b}}
        {2 \ri \mu + \lambda_a + \lambda_b}
        + \frac{\ri(\mu - \nu)}{\ri \mu + \lambda_a} \delta_{a, b},
\label{cA}
\end{split}
\ee
where $a, b \in \bN_n$. As we have seen in \cite{Pusztai_NPB2011}, 
function $\cA$ provides a Lax matrix for the rational $C_n$ RSvD model with 
the two independent parameters $\mu$ and $\nu$. Next, let us consider the
smooth function $\cV \colon \cP^R \rightarrow S$ defined by the equation
\be
    \cV = \cA^{-\half} \cF,
\label{cV}
\ee
and also introduce the product manifold 
\be
    \cM^R = \cP^R \times (K \times K) / U(1)_*,
\label{cM_R}
\ee
where $U(1)_*$ stands for the diagonal embedding of $U(1)$ in the product 
group $K \times K$. Having equipped with the above objects, now we are 
in a position to provide a convenient parametrization of the level set 
$\mfL_0$ (\ref{mfL_0}). Indeed, in \cite{Pusztai_NPB2011} we proved that 
the map
\be
    \Upsilon^R \colon \cM^R \rightarrow \cP^\ext
\label{Upsilon_R_def}
\ee
defined by the assignment
\be
    (\lambda, \theta, (\eta_L, \eta_R) U(1)_*)
    \mapsto
    ( \eta_L \cA(\lambda, \theta)^\half \eta_R^{-1},
        \eta_R \bsLambda(\lambda) \eta_R^{-1},
        \eta_L \xi(\cV(\lambda, \theta)) \eta_L^{-1} )
\label{Upsilon_R}
\ee
is a smooth injective immersion with image $\Upsilon^R(\cM^R) = \mfL_0$. 
Moreover, in \cite{Pusztai_NPB2011} we also proved that $\Upsilon^R$ gives
rise to a diffeomorphism from $\cM^R$ onto the embedded submanifold
$\mfL_0$. In other words, the pair $(\cM^R, \Upsilon^R)$ provides a model 
for the level set $\mfL_0$ (\ref{mfL_0}).

To complete the Marsden--Weinstein reduction, notice that the (residual) 
$K \times K$-action on the model space $\cM^R$ (\ref{cM_R}) takes the form
\be
    (k_L, k_R) \acts (\lambda, \theta, (\eta_L, \eta_R) U(1)_*)
    = (\lambda, \theta, (k_L \eta_L, k_R \eta_R) U(1)_*),
\label{KK_on_cM_R}
\ee
thus the orbit space $\cM^R / (K \times K)$ can be naturally identified with 
the base manifold of the trivial principal $(K \times K) / U(1)_*$-bundle
\be
    \pi^R \colon \cM^R \twoheadrightarrow \cP^R,
    \quad
    (\lambda, \theta, (\eta_L, \eta_R)U(1)_*) 
    \mapsto 
    (\lambda, \theta).
\label{pi_R}
\ee
Now, the crux of the matter is the relationship
\be
    (\pi^R)^* \omega^R = (\Upsilon^R)^* \omega^\ext, 
\label{reduced_symplectic_form}
\ee
that we proved in \cite{Pusztai_NPB2011} by applying a chain of delicate 
arguments. Therefore, for the symplectic quotient in question we obtain 
the identification
\be
    (\cP^\ext \sq_0 (K \times K), \omega^\red) \cong (\cP^R, \omega^R).
\label{symplectic_quotient}
\ee
Finally, note that the $K \times K$-invariant function
\be
    f_1 \colon \cP^\ext \rightarrow \bR,
    \quad
    (y, Y, \rho) \mapsto \half \tr(y y^*)
\label{f_1}
\ee
survives the reduction, and by applying straightforward algebraic 
manipulations one can verify that the corresponding reduced function 
coincides with the Hamiltonian of the rational $C_n$ RSvD system with 
two independent coupling parameters $\mu$ and $\nu$, that can be obtained 
from the $BC_n$-type Hamiltonian (\ref{H_R}) by setting $\kappa = 0$.

\section{Dynamical $r$-matrix for the $C_n$-type model}
\label{SECTION_C_r_matrix}
\setcounter{equation}{0}
Building on the symplectic reduction picture outlined in the previous 
subsection, our goal is to construct a classical $r$-matrix structure for 
the $C_n$-type rational RSvD system with two independent coupling parameters. 
In the context of the CMS models, this geometric approach goes back to the 
work of Avan, Babelon, and Talon \cite{Avan_Babelon_Talon_AA1994}. Eventually, 
in our paper \cite{Pusztai_JMP2012}, we succeeded to construct a dynamical 
$r$-matrix for the most general hyperbolic $BC_n$ Sutherland model with three 
independent coupling constants, too. It is worth mentioning that the 
surrounding ideas proves to be fruitful in the broader context of integrable 
field theories as well. For a systematic review see e.g. 
\cite{Braden_et_al_2003}.

As we have seen in \cite{Pusztai_NPB2011}, the eigenvalues of the Lax matrix
$\cA$ (\ref{cA}) do commute, whence it follows from general principles that 
$\cA$ obeys an $r$-matrix Poisson bracket (for proof, see e.g. 
\cite{Babelon_Bernard_Talon, Babelon_Viallet}). However, we wish to make this 
$r$-matrix structure as explicit as possible. For this reason, Subsection 
\ref{SUBSECTION_local_extension} is devoted to the study of certain local 
extensions for the Lax matrix of the rational $C_n$ RSvD model. As it turns 
out, these local extensions are at the heart of the construction of the 
dynamical $r$-matrix structure for the RSvD system, that we elaborate in 
Subsection \ref{SUBSECTION_computing_r}.

\subsection{Local extensions of the Lax matrix $\cA$}
\label{SUBSECTION_local_extension}
The backbone of our reduction approach is the construction of the so-called 
local extensions of the Lax operator $\cA$ (\ref{cA}), that we wish to 
describe below. For this reason, take an arbitrary point 
\be
    (\lambda^{(0)}, \theta^{(0)}) \in \cP^R
\label{fixed_point_in_cP_R}
\ee
and keep it fixed. Clearly the point
\be
    s^{(0)} 
    = (\lambda^{(0)}, \theta^{(0)}, (\bsone, \bsone) U(1)_*)
        \in \cM^R
\label{s_0}
\ee
is one of the representatives of $(\lambda^{(0)}, \theta^{(0)})$ in $\cM^R$ 
(\ref{cM_R}), that is, $\pi^R(s^{(0)}) = (\lambda^{(0)}, \theta^{(0)})$. 
Moreover, let us introduce the shorthand notations
\be
    \cA^{(0)} = \cA(\lambda^{(0)}, \theta^{(0)}),
    \quad
    \cF^{(0)} = \cF(\lambda^{(0)}, \theta^{(0)}),
    \quad
    \cV^{(0)} = \cV(\lambda^{(0)}, \theta^{(0)}),
\label{cA_0_etc}
\ee
together with
\be
    y^{(0)} = (\cA^{(0)})^\half,
    \quad
    Y^{(0)} = \bsLambda^{(0)} = \bsLambda(\lambda^{(0)}),
    \quad
    \rho^{(0)} = \xi(\cV^{(0)}).
\label{y_0_etc}
\ee
Corresponding to $s^{(0)}$ (\ref{s_0}), in the extended phase space we 
also introduce the reference point
\be
    u^{(0)} 
    = \Upsilon^R (s^{(0)})
    = (y^{(0)}, Y^{(0)}, \rho^{(0)}) 
    \in \cP^\ext.
\label{u_0}
\ee 

Now, associated with the elements given in (\ref{cA_0_etc}-\ref{y_0_etc}), 
let us choose a local section $(e, \sigma)$ of $\pi$ (\ref{pi}), and also a 
local section $W$ of $\xi$ (\ref{xi}), as described in Propositions 
\ref{PROPOSITION_e_and_sigma} and \ref{PROPOSITION_W}, respectively. Upon 
defining the open subset
\be
    \check{\mfg} = \{ Y \in \mfg \, | \, Y_- \in \check{\mfp}_\reg \}
    \subseteq \mfg,
\label{mfg_check}
\ee
it is clear that
\be
    \Psi \colon G \times \check{\mfg} \times \check{\cO} \rightarrow \bC^N,
    \quad
    (y, Y, \rho) \mapsto \sigma(Y_-)^{-1} y^* W(\rho)
\label{Psi}
\ee
is a well-defined smooth function. Due to the conditions imposed in the 
equations (\ref{e_and_sigma_cond}) and (\ref{W_cond}), at the point $u^{(0)}$ 
(\ref{u_0}) for the first $n$ components of $\Psi$ we have 
\be
    \Psi_a(u^{(0)}) = \cF_a^{(0)} \qquad (a \in \bN_n).
\label{Psi_first_n_components}
\ee 
Since these components are strictly positive, there is an open 
subset $\check{\cP}^\ext \subseteq G \times \check{\mfg} \times \check{\cO}$
containing the distinguished point $u^{(0)}$, such that for all $a \in \bN_n$ 
the map
\be
    m_a \colon \check{\cP}^\ext \rightarrow U(1),
    \quad
    u \mapsto \frac{\Psi_a(u)}{\vert \Psi_a(u) \vert}
\label{m_a}
\ee
is well-defined and smooth. Let us keep in mind that by construction
$m_a(u^{(0)}) = 1$.

Now we are in a position to define those group-valued functions that play 
the most important role in the construction of a dynamical $r$-matrix for 
the rational $C_n$ RSvD system. First, making use of the functions $m_a$ 
(\ref{m_a}), we build up the $M$-valued function
\be
    m \colon \check{\cP}^\ext \rightarrow M,
    \quad
    u \mapsto \diag(m_1(u), \ldots, m_n(u), m_1(u), \ldots, m_n(u)),
\label{m_def}
\ee
which satisfies $m(u^{(0)}) = \bsone$. Next, we introduce the $K$-valued 
functions
\begin{align}
    & k \colon \check{\cP}^\ext \rightarrow K,
    \quad
    (y, Y, \rho) \mapsto \sigma(Y_-),
    \label{k_def}
    \\
    & \varphi \colon \check{\cP}^\ext \rightarrow K,
    \quad
    u \mapsto k(u) m(u).
    \label{varphi_def}
\end{align}
Clearly we have $k(u^{(0)}) = \varphi(u^{(0)}) = \bsone$. Finally, we 
define the functions
\begin{align}
    & \bsA^{(0)} \colon \check{\cP}^\ext \rightarrow G,
    \quad
    (y, Y, \rho) \mapsto y^* y,
    \label{bsA_0}
    \\
    & \bsA \colon \check{\cP}^\ext \rightarrow G,
    \quad
    u \mapsto \varphi(u)^{-1} \bsA^{(0)}(u) \varphi(u).
    \label{bsA}
\end{align}
Notice that at the point $u^{(0)}$ (\ref{u_0}) we have 
$\bsA^{(0)}(u^{(0)}) = \bsA(u^{(0)}) = \cA^{(0)}$. Having equipped with the 
above objects, now we can formulate the central result of this subsection.

\begin{LEMMA}
\label{LEMMA_cA_and_bsA}
The $G$-valued smooth function $\bsA$ (\ref{bsA}) is a local extension 
of the Lax matrix $\cA$ (\ref{cA}) around the point $u^{(0)}$ in the 
sense that $\bsA(u^{(0)}) = \cA(\lambda^{(0)}, \theta^{(0)})$ and
\be
    \bsA \circ \Upsilon^R \big|_{(\Upsilon^R)^{-1}(\check{\cP}^\ext)}
    = \cA \circ \pi^R \big|_{(\Upsilon^R)^{-1}(\check{\cP}^\ext)}.
\label{cA_and_bsA}
\ee
\end{LEMMA}

\begin{proof}
It is enough to verify (\ref{cA_and_bsA}). For, take an arbitrary point
\be
    s =(\lambda, \theta, (\eta_L, \eta_R) U(1)_*) 
        \in (\Upsilon^R)^{-1}(\check{\cP}^\ext)
\label{L3_point_s}
\ee
with some $\lambda \in \mfc$, $\theta \in \bR^n$ and $\eta_L, \eta_R \in K$.
Also, for brevity we define
\be
    u = (y, Y, \rho) = \Upsilon^R (s) \in \check{\cP}^\ext.
\label{L3_u}
\ee
Recalling (\ref{pi}) and (\ref{Upsilon_R}), it is clear that
\be
    Y = \eta_R \bsLambda(\lambda) \eta_R^{-1} =  \pi(\lambda, \eta_R) 
        \in \mfp_\reg.
\label{L3_pi_and_Y}
\ee
On the other hand, since $Y \in \check{\mfg}$, we have 
$Y_- \in \check{\mfp}_\reg$. Thus, making use of the local section 
$(e, \sigma)$ introduced in (\ref{e_and_sigma}), we see that
\be
    Y = Y_- = \pi(e(Y_-), \sigma(Y_-)) = \pi(e(Y_-), k(u))
\label{L3_Y_and_section}
\ee
also holds. Recalling (\ref{M_action}), the comparison of (\ref{L3_pi_and_Y}) 
and (\ref{L3_Y_and_section}) yields that there is a unique element
\be
    \tilde{m} 
    = \diag(\tilde{m}_1, \ldots, \tilde{m}_n, 
        \tilde{m}_1, \ldots, \tilde{m}_n) 
        \in M
\label{L3_tilde_m}
\ee 
such that
\be
    (\lambda, \eta_R) = (e(Y_-), k(u) \tilde{m}).
\label{L3_m_tilde}
\ee

Next, remembering the parametrization (\ref{Upsilon_R}), we can write
\be
    \xi(\eta_L \cV(\lambda, \theta))
    = \eta_L \xi(\cV(\lambda, \theta)) \eta_L^{-1}
    = \rho
    \in \check{\cO}.
\label{L3_cV_and_W}
\ee
However, utilizing the local section $W$ introduced in (\ref{W}), we 
also have $\xi(W(\rho)) = \rho$, whence by (\ref{U(1)_action_on_S}) we can 
write that
\be
    \eta_L \cV(\lambda, \theta) = e^{\ri \psi} W(\rho)
\label{L3_eta_L_and_cV_and_W}
\ee
with some constant $\psi \in \bR$. From the above observations it readily 
follows that
\be
    e^{\ri \psi} W(\rho)
    = \eta_L \cA(\lambda, \theta)^{-\half} \cF(\lambda, \theta)
    = (y^*)^{-1} \eta_R \cF(\lambda, \theta)
    = (y^*)^{-1} k(u) \tilde{m} \cF(\lambda, \theta),
\label{L3_key_rel}
\ee
from where we get $e^{\ri \psi} \Psi(u) = \tilde{m} \cF(\lambda, \theta)$.
Componentwise, for each $a \in \bN_n$ we can write
\be
    e^{\ri \psi} \Psi_a(u) = \tilde{m}_a \cF_a(\lambda, \theta),
\label{L3_Phi_components}
\ee
thus the relationship $\vert \Psi_a(u) \vert = \cF_a(\lambda, \theta)$ and 
also
\be
    \tilde{m}_a 
    = e^{\ri \psi} \frac{\Psi_a(u)}{\cF_a(\lambda, \theta)}
    = e^{\ri \psi} \frac{\Psi_a(u)}{\vert \Psi_a(u) \vert} 
    = e^{\ri \psi} m_a(u)
\label{L3_m_and_tilde_m}
\ee
are evident. In other words, $\tilde{m} = e^{\ri \psi} m(u)$, whence
from (\ref{L3_m_tilde}) and (\ref{varphi_def}) we conclude that
\be
    \eta_R 
    = k(u) \tilde{m} 
    = e^{\ri \psi} k(u) m(u) 
    = e^{\ri \psi} \varphi(u).
\label{L3_eta_R_and_varpsi}
\ee 

Now, turning to the functions (\ref{bsA_0}) and (\ref{bsA}), notice that
\be
    \bsA^{(0)}(u) = y^* y = \eta_R \cA(\lambda, \theta) \eta_R^{-1},
\label{L3_bsA_0}
\ee
so from (\ref{L3_eta_R_and_varpsi}) we infer that
\be
    \bsA \circ \Upsilon^R (s) 
    = \bsA(u) 
    = \varphi(u)^{-1} \bsA^{(0)}(u) \varphi(y) 
    = \cA(\lambda, \theta)
    = \cA \circ \pi^R (s).
\label{L3_bsA_and_cA}
\ee
Since $s$ (\ref{L3_point_s}) is an arbitrary element of 
$(\Upsilon^R)^{-1}(\check{\cP}^\ext)$, the Lemma follows.
\end{proof}

\subsection{Computing the $r$-matrix}
\label{SUBSECTION_computing_r}
The natural idea impregnated by Lemma \ref{LEMMA_cA_and_bsA} is that the 
Poisson brackets of the components of the Lax matrix $\cA$ (\ref{cA}) can be
computed by inspecting the Poisson brackets of the components of the locally 
defined function $\bsA$ (\ref{bsA}). Indeed, since we reduce the symplectic 
manifold $\cP^\ext$ (\ref{cP_ext}) at the zero value of the 
$K \times K$-equivariant momentum map $J^\ext$ (\ref{J_ext}), and since the 
local extension $\bsA$ is (locally) $K \times K$-invariant on the level set 
$\mfL_0$ (\ref{mfL_0}), using the St. Petersburg tensorial notation we can 
simply write
\be
    \{ \cA \stackrel{\otimes}{,} \cA \}^R (\lambda^{(0)}, \theta^{(0)}) 
    = \{ \bsA \stackrel{\otimes}{,} \bsA \}^\ext(u^{(0)}).
\label{cA_and_bsA_PB} 
\ee
However, for the function $\bsA^{(0)}$ (\ref{bsA_0}) we clearly have
\be
    \{ \bsA^{(0)} \stackrel{\otimes}{,} \bsA^{(0)} \}^\ext = 0,
\label{bsA_0_PB}
\ee
that is, $\bsA^{(0)}$ obeys an $r$-matrix bracket with the trivial zero 
$r$-matrix. Therefore, due to the relationship 
$\bsA = \varphi^{-1} \bsA^{(0)} \varphi$ (\ref{bsA}), it is clear that $\bsA$
also obeys a linear $r$-matrix bracket
\be
    \{ \bsA \stackrel{\otimes}{,} \bsA \}^\ext
    = [\tilde{r}_{12}, \bsA \otimes \bsone] 
        - [\tilde{r}_{21}, \bsone \otimes \bsA]  
\label{bsA_PB}
\ee
with the transformed $r$-matrix
\be
    \tilde{r}_{12} 
    = \varphi_1^{-1} \varphi_2^{-1} 
        \left(
            -\{ \varphi_1, \bsA^{(0)}_2 \}^\ext \varphi_1^{-1}
            + \half [ \{ \varphi_1, \varphi_2\}^\ext 
                            \varphi_1^{-1} \varphi_2^{-1},
                    \bsA^{(0)}_2 ]
        \right)
    \varphi_1 \varphi_2.
\label{r_tilde}
\ee
Now, recalling that $\varphi(u^{(0)}) = \bsone$, from the relationships 
(\ref{cA_and_bsA_PB}) and (\ref{r_tilde}) we infer that for the Lax matrix 
$\cA$ we can write
\be
    \{ \cA_1, \cA_2 \}^R (\lambda^{(0)}, \theta^{(0)}) 
    = [r_{12}(\lambda^{(0)}, \theta^{(0)}), 
        \cA_1(\lambda^{(0)}, \theta^{(0)})] 
    - [r_{21}(\lambda^{(0)}, \theta^{(0)}), 
        \cA_2(\lambda^{(0)}, \theta^{(0)})]
\label{cA_PB}
\ee
with the $r$-matrix 
\be
    r_{12}(\lambda^{(0)}, \theta^{(0)}) 
    = -\{ \varphi_1, \bsA^{(0)}_2 \}^\ext(u^{(0)})
     + \half [ \{ \varphi_1, \varphi_2\}^\ext(u^{(0)}), 
                \cA_2(\lambda^{(0)}, \theta^{(0)}) ].
\label{r}
\ee
However, since $\varphi = k m$ (\ref{varphi_def}), Leibniz rule yields
\be
    \{ \varphi_1, \bsA^{(0)}_2 \}^\ext (u^{(0)})
    = \{ k_1, \bsA^{(0)}_2 \}^\ext (u^{(0)})
        + \{ m_1, \bsA^{(0)}_2 \}^\ext (u^{(0)}),
\label{varphi_bsA_0_PB}
\ee
together with
\be
\begin{split}
    \{ \varphi_1, \varphi_2 \}^\ext (u^{(0)})
    = & \{ k_1, k_2 \}^\ext (u^{(0)}) 
    + \{ k_1, m_2 \}^\ext (u^{(0)}) \\
    & + \{ m_1, k_2 \}^\ext (u^{(0)}) 
    + \{ m_1, m_2 \}^\ext (u^{(0)}).
\end{split}
\label{varphi_varphi_PB}
\ee
Thus, in order to provide an explicit formula for the above $r$-matrix 
(\ref{r}), we still have to work out the Poisson brackets appearing in 
(\ref{varphi_bsA_0_PB}) and (\ref{varphi_varphi_PB}). However, recalling 
(\ref{PB_ext}), it essentially boils down to the computation of the gradients 
(\ref{F_gradients}) of the components of the functions $\bsA^{(0)}$, $k$ and 
$m$. We accomplish these tasks in the following series of short Propositions.

\begin{PROPOSITION}
\label{PROPOSITION_bsA_0_gradients}
At the reference point $u^{(0)}$ (\ref{u_0}), for all matrix 
$v \in \mfgl(N, \bC)$ we have
\begin{align}
    & (\nabla^G \Real (\tr ( v \bsA^{(0)} )))(u^{(0)}) 
    = \frac{ (v + v^*) \cA^{(0)} + (\cA^{(0)})^{-1} (v + v^*) }{2}, 
    \label{G_grad_bsA_0_Re}
    \\
    & (\nabla^G \Imag (\tr ( v \bsA^{(0)} )))(u^{(0)}) 
    = \frac{ (v - v^*) \cA^{(0)} - (\cA^{(0)})^{-1} (v - v^*) }{2 \ri},
    \label{G_grad_bsA_0_Im}
\end{align}
whereas the remaining gradients of $\Real (\tr ( v \bsA^{(0)} ))$ and 
$\Imag (\tr ( v \bsA^{(0)} ))$ are trivial, i.e.
\begin{align}
    & (\nabla^\mfg \Real (\tr ( v \bsA^{(0)} )))(u^{(0)})
    = (\nabla^\mfg \Imag (\tr ( v \bsA^{(0)} )))(u^{(0)})
    = 0,
    \label{mfg_grad_bsA_0}
    \\
    & (\nabla^\cO \Real (\tr ( v \bsA^{(0)} )))(u^{(0)}) 
    = (\nabla^\cO \Imag (\tr ( v \bsA^{(0)} )))(u^{(0)})
    = 0.
    \label{cO_grad_bsA_0}
\end{align}
\end{PROPOSITION}

\begin{proof}
Take an arbitrary tangent vector 
\be
    \Delta u
    = \delta y \oplus \delta Y \oplus [X, \rho^{(0)}] 
        \in T_{u^{(0)}} \check{\cP}^\ext
\label{Delta_u}
\ee
with some Lie algebra element $X \in \mfk$. By neglecting the second and the 
higher order terms in the small real parameter $t$, one can easily find that
\be
\begin{split}
    & \bsA^{(0)}(u^{(0)} + t \Delta u + \cdots) \\
    & \quad 
    = \cA^{(0)} 
        + t \left( 
            \cA^{(0)} ((y^{(0)})^{-1} \delta y) 
            + ((y^{(0)})^{-1} \delta y)^* \cA^{(0)} 
        \right)
        + \cdots,
\end{split}
\label{P4_bsA_0_expansion}
\ee
from where we infer that
\be
\begin{split}
    & \Real (\tr ( v \bsA^{(0)} ))_{* u^{(0)}} 
        (\Delta u)  
    = \left\{
        \frac{\dd}{\dd t}
        \Real (\tr ( v \bsA^{(0)} (u^{(0)} + t \Delta u + \cdots))) 
    \right\}_{t = 0} \\
    & \quad
    = \tr \left(
        \frac{(v + v^*) \cA^{(0)} + (\cA^{(0)})^{-1} (v + v^*)}{2} 
            (y^{(0)})^{-1} \delta y
    \right).
\end{split}
\label{P4_bsA_0_real}
\ee
Similarly, one obtains immediately that
\be
\begin{split}
    & \Imag (\tr ( v \bsA^{(0)} ))_{* u^{(0)}} 
        (\Delta u) 
    = \tr \left(
        \frac{(v - v^*) \cA^{(0)} - (\cA^{(0)})^{-1} (v - v^*)}{2 \ri}
            (y^{(0)})^{-1} \delta y
    \right).
\end{split}
\label{P4_bsA_0_imag}
\ee
Since the $N \times N$ matrices appearing on the right hand side of both 
(\ref{G_grad_bsA_0_Re}) and (\ref{G_grad_bsA_0_Im}) do belong to the Lie 
algebra $\mfg$ (\ref{mfg}), by the definition of the gradients
(\ref{F_gradients_def}) the Proposition follows.
\end{proof}

\begin{PROPOSITION}
\label{PROPOSITION_k_gradients}
At the distinguished point $u^{(0)}$ (\ref{u_0}), for all Lie algebra
element $v \in \mfg$ we have
\begin{align}
    (\nabla^\mfg \Real (\tr ( v k)))(u^{(0)}) 
    = \half (\wad_{\bsLambda^{(0)}})^{-1} 
        \left( \left( v - v^* \right)_{\mfm^\perp} \right) 
    \quad \text{and} \quad
    (\nabla^\mfg \Imag (\tr ( v k )))(u^{(0)}) = 0,
    \label{mfg_grad_k}
\end{align}
while the remaining gradients of $\Real (\tr ( v k))$ and 
$\Imag (\tr ( v k))$ are trivial, i.e.
\begin{align}
    & (\nabla^G \Real (\tr ( v k )))(u^{(0)}) 
    = (\nabla^G \Imag (\tr ( v k )))(u^{(0)})
    = 0,
    \label{G_grad_k}
    \\
    & (\nabla^\cO \Real (\tr ( v k )))(u^{(0)}) 
    = (\nabla^\cO \Imag (\tr ( v k )))(u^{(0)})
    = 0.
    \label{cO_grad_k}
\end{align}
\end{PROPOSITION}

\begin{proof}
As in the proof of Proposition \ref{PROPOSITION_bsA_0_gradients}, take 
an arbitrary tangent vector $\Delta u$ as given in (\ref{Delta_u}). By 
applying a first order expansion on $k$ (\ref{k_def}) in the small real 
parameter $t$, Proposition \ref{PROPOSITION_e_and_sigma} tells us that
\be
    k(u^{(0)} + t \Delta u + \cdots)
    = \sigma( (\bsLambda^{(0)} + t \delta Y)_- )
    = \sigma( \bsLambda^{(0)} + t (\delta Y)_- )
    = \bsone + t \delta \sigma + \cdots,
\label{P5_k_expansion}
\ee
where
\be
    \delta \sigma 
    = - (\wad_{\bsLambda^{(0)}})^{-1} 
            ((\delta Y)_{\mfa^\perp}) \in \mfm^\perp.
\label{P5_delta_sigma}
\ee
Since for all $v \in \mfg$ we have
\be
    \Real (\tr (v k)) = \frac{\tr(v k) + \tr(v^* k^{-1})}{2}
    \quad \text{and} \quad
    \Imag (\tr (v k)) = \frac{\tr(v k) - \tr(v^* k^{-1})}{2 \ri},
\label{P5_on_tr}
\ee
it is now evident that
\be
    \Real( \tr ( v k(u^{(0)} + t \Delta u + \cdots))) 
    = \frac{\tr(v) + \tr(v^*)}{2} 
        + t \tr \left( \frac{v - v^*}{2} \delta \sigma \right)
        + \cdots.
\label{P5_real_tr}
\ee
Recalling the Cartan involution (\ref{vartheta}) we can write 
$v - v^* = v + \vartheta(v) \in \mfk$, therefore
\be
\begin{split}
    & (\Real (\tr(v k)))_{* u^{(0)}} 
        (\Delta u)
    = - \half \left \langle
        v + \vartheta(v), 
        (\wad_{\bsLambda^{(0)}})^{-1} ((\delta Y)_{\mfa^\perp})
    \right \rangle \\
    & \quad
    = \half \left \langle
        (\wad_{\bsLambda^{(0)}})^{-1} 
            \left( \left( v - v^* \right)_{\mfm^\perp} \right),
        \delta Y   
    \right \rangle.
\end{split}
\label{P5_real_tr_der}
\ee
In complete analogy with (\ref{P5_real_tr}), for the function 
$\Imag(\tr(v k))$ we can write the expansion
\be
    \Imag( \tr ( v k(u^{(0)} + t \Delta u + \cdots))) 
    = \frac{\tr(v) - \tr(v^*)}{2 \ri} 
        + t \tr \left( \frac{v + v^*}{2 \ri} \delta \sigma \right)
        + \cdots.
\label{P5_imag_tr}
\ee
However, since $v + v^* = v - \vartheta(v) \in \mfp$, and since the subspaces 
$\mfk$ and $\mfp$ (\ref{mfkp}) are orthogonal with respect to the bilinear 
form (\ref{bilinear_form}), we conclude that
\be
    (\Imag (\tr(v k)))_{* u^{(0)}} (\Delta u)
    = \frac{1}{2 \ri} 
        \left \langle 
            v - \vartheta(v), \delta \sigma
        \right \rangle
    = 0,
\label{P5_imag_tr_der}
\ee
thus by recalling (\ref{F_gradients_def}) the Proposition follows.
\end{proof}

To find the gradients of the components of $m$ (\ref{m_def}) we find it 
convenient to introduce the auxiliary function
\be
    \tau \colon \bC \setminus \{ 0 \} \rightarrow U(1),
    \quad
    z \mapsto \frac{z}{\vert z \vert}.
\label{tau}
\ee
It is clearly smooth, and at each point $x > 0$ for its derivative we have
\be
    \tau_{* x}(w) = \ri \frac{\Imag(w)}{x}
    \qquad
    (w \in \bC \cong T_x (\bC \setminus \{ 0 \})).
\label{tau_derivative}
\ee
Note that with the aid of $\tau$ the function $m_a$ (\ref{m_a}) can be 
simply written as 
\be
    m_a = \tau \circ \Psi_a.
\label{m_and_tau}
\ee 
Also, for each $a \in \bN_n$ we introduce the $N \times N$ matrix
\be
    \zeta_a 
    = \ri \frac{e_a (\cF^{(0)})^* + \bsC (\cF^{(0)}) e_a^* \bsC}{2},
\label{zeta_a}
\ee
where $e_k \in \bC^N$ denotes the column vector with components
\be
    (e_k)_l = \delta_{k, l}
    \qquad
    (k, l \in \bN_N).
\label{e_vectors}
\ee 
As a matter of fact, the above matrix $\zeta_a$ belongs to the Lie algebra 
$\mfg$. Moreover, utilizing the basis (\ref{v_I}), we can write
\be
\begin{split}
    2 \zeta_a 
    = & \sqrt{2} \cF_a^{(0)} D_a^+ 
        - \sqrt{2} \Real(\cF_{n + a}^{(0)}) 
            (X^{+, \ri}_{2 \veps_a} + X^{-, \ri}_{2 \veps_a}) \\
    & + \sum_{c = 1}^{a - 1}
        \left(
            \cF_c^{(0)} 
                (X^{+, \ri}_{\veps_c - \veps_a} 
                - X^{-, \ri}_{\veps_c - \veps_a})
            - \Real(\cF_{n + c}^{(0)}) 
                (X^{+, \ri}_{\veps_c + \veps_a} 
                + X^{-, \ri}_{\veps_c + \veps_a}) 
        \right. \\
    & \qquad \qquad 
        \left. 
        + \Imag(\cF_{n + c}^{(0)}) 
            (X^{+, \rr}_{\veps_c + \veps_a} 
            + X^{-,\rr}_{\veps_c + \veps_a})
        \right) \\
    & + \sum_{c = a + 1}^n
        \left(
            \cF_c^{(0)} 
                (X^{+, \ri}_{\veps_a - \veps_c} 
                + X^{-, \ri}_{\veps_a - \veps_c})
            - \Real(\cF_{n + c}^{(0)}) 
                (X^{+, \ri}_{\veps_a + \veps_c} 
                + X^{-, \ri}_{\veps_a + \veps_c}) 
        \right. \\
    & \qquad \qquad 
        \left. 
        - \Imag(\cF_{n + c}^{(0)}) 
            (X^{+, \rr}_{\veps_a + \veps_c} 
            + X^{-,\rr}_{\veps_a + \veps_c})
        \right).
\end{split}
\label{zeta_a_expanded}
\ee

\begin{PROPOSITION}
\label{PROPOSITION_m_gradients}
Take an arbitrary $a \in \bN_n$, then at the reference point $u^{(0)}$ 
(\ref{u_0}) the gradients (\ref{F_gradients}) of the function $\Real (m_a)$ 
are all zeros. However, for the imaginary part of the function $m_a$ 
(\ref{m_a}) we have the non-trivial formulae
\begin{align}
    & (\nabla^G (\Imag(m_a)))(u^{(0)}) = \frac{1}{\cF_a^{(0)}} \zeta_a
        \in \mfg,
    \label{G_grad_imag_m_a} \\
    & (\nabla^\mfg (\Imag(m_a)))(u^{(0)}) 
    = \frac{1}{\cF_a^{(0)}} 
        (\wad_{\bsLambda^{(0)}})^{-1}((\zeta_a)_{\mfm^\perp}) 
        \in \mfa^\perp,
    \label{mfg_grad_imag_m_a} \\
    & (\nabla^\cO (\Imag(m_a)))(u^{(0)}) 
    = \xi_{* \cV^{(0)}} (\delta V_a) \in T_{\rho^{(0)}} \cO,
    \label{cO_grad_imag_m_a}
\end{align}
where
\be
    \delta V_a 
    = \frac{\bsC (\cA^{(0)})^\half e_a - (\cA^{(0)})^\half e_a}
        {4 \mu \cF_a^{(0)}}
        + \frac{\cV^{(0)}}{2 \mu N} 
        \in T_{\cV^{(0)}} S.
\label{delta_V_a}
\ee
\end{PROPOSITION}

\begin{proof}
First, take an arbitrary tangent vector $\delta y \in T_{y^{(0)}} G$. 
Recalling (\ref{Psi}), and the conditions (\ref{e_and_sigma_cond}), 
(\ref{W_cond}), for small values of the real parameter $t$ we can write 
the first order expansion
\be
\begin{split}
    & \Psi(y^{(0)} + t \delta y + \cdots, Y^{(0)}, \rho^{(0)}) \\
    & \quad
    = \sigma((Y^{(0)})_-)^{-1} (y^{(0)} + t \delta y + \cdots)^* W(\rho^{(0)})
    = \cF^{(0)} - t \bsC (y^{(0)})^{-1} (\delta y) \bsC \cF^{(0)} + \cdots,
\end{split}
\label{P6_Psi_expansion_G}
\ee
thus for each $a \in \bN_n$ we have
\be
    \Psi_a (y^{(0)} + t \delta y + \cdots, Y^{(0)}, \rho^{(0)})
    = \cF_a^{(0)} + t w_a + \cdots   
\label{P6_Psi_a_expansion_G}
\ee
with 
$w_a = - \tr(\bsC (y^{(0)})^{-1} (\delta y) \bsC \cF^{(0)} e_a^*) \in \bC$.
Recalling (\ref{tau_derivative}), we can write
\be
    m_a (y^{(0)} + t \delta y + \cdots, Y^{(0)}, \rho^{(0)})
    = \tau (\cF_a^{(0)} + t w_a + \cdots)
    = 1 + t \ri \frac{\Imag(w_a)}{\cF_a^{(0)}} + \cdots,
\label{P6_m_a_expansion_G}
\ee
thus clearly $(\nabla^G (\Real(m_a)))(u^{(0)}) = 0$. Moreover, by inspecting 
$w_a$ and (\ref{zeta_a}), we obtain
\be
    (\Imag(m_a))_{* u^{(0)}} (\delta y \oplus 0 \oplus 0) 
    = \frac{\Imag(w_a)}{\cF_a^{(0)}} 
    = \frac{\tr( \zeta_a (y^{(0)})^{-1} \delta y)}{\cF_a^{(0)}},
\label{P6_m_a_der_G}
\ee
so (\ref{G_grad_imag_m_a}) also follows immediately.

Second, take an arbitrary tangent vector 
$\delta Y \in \mfg \cong T_{Y^{(0)}} \mfg$. According to Proposition 
\ref{PROPOSITION_e_and_sigma}, for small values of $t \in \bR$ we have 
the first order expansion
\be
    \sigma((Y^{(0)} + t \delta Y)_-)
    = \sigma(\bsLambda^{(0)} + t (\delta Y)_-)
    = \bsone + t \delta \sigma + \cdots,
\label{P6_k_expansion_mfg}
\ee
with the Lie algebra element $\delta \sigma \in \mfm^\perp$ displayed in 
(\ref{P5_delta_sigma}). Therefore,
\be
    \Psi(y^{(0)}, Y^{(0)} + t \delta Y, \rho^{(0)})
    = \sigma((Y^{(0)} + t \delta Y)_-)^{-1} (y^{(0)})^* W(\rho^{(0)})
    = \cF^{(0)} - t (\delta \sigma) \cF^{(0)} + \cdots,
\label{P6_Psi_expansion_mfg}
\ee
and so for each $a \in \bN_n$ we can write
\be
    \Psi_a (y^{(0)}, Y^{(0)} + t \delta Y, \rho^{(0)})
    = \cF_a^{(0)} + t w'_a + \cdots   
\label{P6_Psi_a_expansion_mfg}
\ee
with $w'_a = - \tr((\delta \sigma) \cF^{(0)} e_a^*) \in \bC$. Utilizing 
$\tau$ (\ref{tau}) and its derivative (\ref{tau_derivative}), we obtain
\be
    m_a (y^{(0)}, Y^{(0)} + t \delta Y, \rho^{(0)})
    = \tau (\cF_a^{(0)} + t w'_a + \cdots)
    = 1 + t \ri \frac{\Imag(w'_a)}{\cF_a^{(0)}} + \cdots,
\label{P6_m_a_expansion_mfg}
\ee
therefore $(\nabla^\mfg (\Real(m_a)))(u^{(0)}) = 0$ is immediate. 
Remembering (\ref{zeta_a}) it is also clear that
\be
    \Imag(w'_a) 
    = \langle \zeta_a, \delta \sigma \rangle
    = - \left \langle 
            (\zeta_a)_{\mfm^\perp}, 
            (\wad_{\bsLambda^{(0)}})^{-1} ((\delta Y)_{\mfa^\perp})
    \right \rangle
    = \left \langle
        (\wad_{\bsLambda^{(0)}})^{-1} ((\zeta_a)_{\mfm^\perp}), \delta Y
    \right \rangle.  
\label{P6_Im_w'_a_mfg}
\ee
Thus, by combining (\ref{P6_m_a_expansion_mfg}) and (\ref{P6_Im_w'_a_mfg}), 
we end up with the formula
\be
    (\Imag(m_a))_{* u^{(0)}} (0 \oplus \delta Y \oplus 0) 
    = \frac{\Imag(w'_a)}{\cF_a^{(0)}} 
    = \left \langle
        \frac{1}{\cF_a^{(0)}}(\wad_{\bsLambda^{(0)}})^{-1} 
            ((\zeta_a)_{\mfm^\perp}), 
        \delta Y
    \right \rangle,
\label{P6_Imag_m_a_der_G}
\ee
that readily implies (\ref{mfg_grad_imag_m_a}).

Third, take an arbitrary $X \in \mfk$. Remembering (\ref{e_and_sigma_cond}), 
(\ref{W_cond}), and (\ref{Psi}), notice that
\be
\begin{split}
    & \Psi(y^{(0)}, Y^{(0)}, \rho^{(0)} + t [X, \rho^{(0)}] + \cdots) \\
    & \quad
    = \cF^{(0)} 
        + t \left(
            (\cA^{(0)})^\half  X \cV^{(0)} 
                - \frac{(\cV^{(0)})^* X \cV^{(0)}}{N} \cF^{(0)}
        \right)
        + \cdots, 
\end{split}
\label{P6_Psi_expansion_cO}
\ee
whence for all $a \in \bN_n$ we can write
\be
    \Psi_a (y^{(0)}, Y^{(0)}, \rho^{(0)} + t [X, \rho^{(0)}] + \cdots)
    = \cF^{(0)}_a + t w''_a + \cdots
\label{P6_Psa_a_expansion}
\ee
with the complex number
\be
    w''_a 
    = \tr((\cA^{(0)})^\half  X \cV^{(0)} e_a^*) 
        - \frac{\cF^{(0)}_a}{N} \tr(X \cV^{(0)} (\cV^{(0)})^*).
\label{P6_w''_a}
\ee
It readily follows that
\be
    m_a (y^{(0)}, Y^{(0)}, \rho^{(0)} + t [X, \rho^{(0)}] + \cdots)
    = \tau(\cF^{(0)}_a + t w''_a + \cdots)
    = 1 + t \ri \frac{\Imag(w''_a)}{\cF^{(0)}_a} + \cdots,
\label{P6_m_a_expansion}
\ee
from where we get at once that $(\nabla^\cO (\Real(m_a)))(u^{(0)}) = 0$ and
\be
    (\Imag(m_a))_{* u^{(0)}} (0 \oplus 0 \oplus [X, \rho^{(0)}]) 
    = \frac{\Imag(w''_a)}{\cF^{(0)}_a}.
\label{P6_Imag_m_a_grad_cO}
\ee
At this point notice that $\ri \cV^{(0)} (\cV^{(0)})^* \in \mfk$. Therefore, 
recalling (\ref{zeta_a}), we can write
\be
\begin{split}
    & \Imag(w''_a) 
    = - \tr((\cA^{(0)})^\half \zeta_a (\cA^{(0)})^{-\half} X)
        + \frac{\cF^{(0)}_a}{N} \tr( \ri \cV^{(0)} (\cV^{(0)})^* X) \\
    & \quad
    = \cF^{(0)}_a 
        \left \langle
            - \frac{1}{\cF^{(0)}_a} 
                \left(
                    (\cA^{(0)})^\half \zeta_a (\cA^{(0)})^{-\half}
                \right)_{\mfk}
            + \frac{\ri}{N} \cV^{(0)} (\cV^{(0)})^*, X
        \right \rangle.
\end{split}
\label{P6_Imag_w''_a}
\ee
Now, one can verify that the column vector $\delta V_a$ displayed in 
(\ref{delta_V_a}) does belong to the tangent space $T_{\cV^{(0)}} S$ 
(\ref{TS}). Furthermore, recalling (\ref{xi_der}) we find that
\be
    - \frac{1}{\cF^{(0)}_a} 
        \left(
            (\cA^{(0)})^\half \zeta_a (\cA^{(0)})^{-\half}
        \right)_{\mfk}
    + \frac{\ri}{N} \cV^{(0)} (\cV^{(0)})^*
    = \xi_{* \cV^{(0)}} (\delta V_a),
\label{P6_xi_enters}
\ee
thus the relationship (\ref{cO_grad_imag_m_a}) also follows.
\end{proof}

Having the necessary gradients at our disposal, now we are ready to work
out the tensorial Poisson brackets appearing in (\ref{varphi_bsA_0_PB}) 
and (\ref{varphi_varphi_PB}).

\begin{LEMMA}
\label{LEMMA_PBs_with_bsA_0}
At the point $u^{(0)}$ (\ref{u_0}) we can write
\begin{align}
    & \{ k \stackrel{\otimes}{,} \bsA^{(0)} \}^\ext (u^{(0)})
    = \cA^{(0)}_2
    \left(
        \sum_{\alpha, \eps}
            \frac{X^{+, \eps}_\alpha 
                    \otimes X^{-, \eps}_\alpha}{\alpha(\lambda^{(0)})}
    \right)
    + \left(
        \sum_{\alpha, \eps}
            \frac{X^{+, \eps}_\alpha 
                    \otimes X^{-, \eps}_\alpha}{\alpha(\lambda^{(0)})}
        \right) \cA^{(0)}_2,
    \label{PB_k_and_bsA_0} \\
    & \{ m \stackrel{\otimes}{,} \bsA^{(0)} \}^\ext (u^{(0)})
    = - \cA^{(0)}_2 \left( \sum_{a = 1}^n D^+_a \otimes S^{(0)}_a \right) 
        - \left( \sum_{a = 1}^n D^+_a \otimes S^{(0)}_a \right) \cA^{(0)}_2,
    \label{PB_m_and_bsA_0} 
\end{align}
where for each $a \in \bN_n$ we have
\be
\begin{split}
    S^{(0)}_a
    = \frac{1}{\sqrt{2} \cF^{(0)}_a}
        & \left\{
            - \frac{\sqrt{2} \Real(\cF^{(0)}_{n + a}) 
                    X^{-, \ri}_{2 \veps_a}}
                {2 \lambda^{(0)}_a}
        \right. \\
    & \quad 
    + \sum_{c = 1}^{a - 1} 
        \left(
            \frac{\cF^{(0)}_c X^{-, \ri}_{\veps_c - \veps_a}}
                {\lambda^{(0)}_c - \lambda^{(0)}_a} 
            - \frac{\Real(\cF^{(0)}_{n + c}) 
                    X^{-, \ri}_{\veps_c + \veps_a}}
                {\lambda^{(0)}_c + \lambda^{(0)}_a}
            + \frac{\Imag(\cF^{(0)}_{n + c}) 
                    X^{-, \rr}_{\veps_c + \veps_a}}
                {\lambda^{(0)}_c + \lambda^{(0)}_a}
        \right) \\
    & \left. \quad 
    + \sum_{c = a + 1}^n 
        \left(
            \frac{\cF^{(0)}_c X^{-, \ri}_{\veps_a - \veps_c}}
                {\lambda^{(0)}_a - \lambda^{(0)}_c} 
            - \frac{\Real(\cF^{(0)}_{n + c}) 
                    X^{-, \ri}_{\veps_a + \veps_c}}
                {\lambda^{(0)}_a + \lambda^{(0)}_c}
            - \frac{\Imag(\cF^{(0)}_{n + c}) 
                    X^{-, \rr}_{\veps_a + \veps_c} }
                {\lambda^{(0)}_a + \lambda^{(0)}_c}
        \right)
    \right\}.
\end{split}
\label{S_a}
\ee
\end{LEMMA}

\begin{proof}
To prove (\ref{PB_k_and_bsA_0}), we utilize the family of $N \times N$ 
matrices $\{ v_I \}$ defined in (\ref{v_I}), that forms a basis in the 
complex linear space $\mfgl(N, \bC)$. Recalling Proposition 
\ref{PROPOSITION_k_gradients}, we see that among the members of the
basis $\{ v_I \}$ only the vectors $X^{+,\eps}_\alpha$ 
$(\alpha \in \cR_+, \eps \in \{ \pm \})$ generate non-trivial gradients 
of the form
\be
    \nabla^\mfg (\Real(\tr(X^{+,\eps}_\alpha k)))(u^{(0)})
    = \frac{1}{\alpha(\lambda^{(0)})} X^{-, \eps}_\alpha.
\label{L7_k_nontrivial_gradients}
\ee
Let $\{ v^I \} \subseteq \mfgl(N, \bC)$ denote the dual basis of 
$\{ v_I \}$ provided by the trace-pairing of $\mfgl(N, \bC)$. Recalling 
Proposition \ref{PROPOSITION_bsA_0_gradients} and the explicit expression
of the gradients (\ref{L7_k_nontrivial_gradients}), the Poisson bracket 
formula (\ref{PB_ext}) allows us to write
\be
\begin{split}
    & \{ k \stackrel{\otimes}{,} \bsA^{(0)} \}^\ext (u^{(0)}) 
    = \sum_{I, J} \{ \tr(v_I k), \tr(v_J \bsA^{(0)}) \}^\ext (u^{(0)})
        v^I \otimes v^J \\
    & \quad
    = \sum_{\alpha, \eps} \sum_J
        \frac{\tr(v_J \cA^{(0)} X^{-, \eps}_\alpha)
                - \tr(v_J \bsC X^{-, \eps}_\alpha \bsC \cA^{(0)})}
        {\alpha(\lambda^{(0)})}
        X^{+, \eps}_\alpha \otimes v^J \\
    & \quad
    = \sum_{\alpha, \eps} 
        \frac{1}{\alpha(\lambda^{(0)})}
        X^{+, \eps}_\alpha 
            \otimes 
            (\cA^{(0)} X^{-, \eps}_\alpha 
                + X^{-, \eps}_\alpha \cA^{(0)}),
\end{split}
\label{L7_PB_k_and_bsA_0}
\ee
so (\ref{PB_k_and_bsA_0}) follows immediately. 

Making use of Propositions 
\ref{PROPOSITION_bsA_0_gradients} and \ref{PROPOSITION_m_gradients}, let 
us notice that the Poisson bracket formula (\ref{PB_ext}) yields
\be
\begin{split}
    & \{ m \stackrel{\otimes}{,} \bsA^{(0)} \}^\ext (u^{(0)}) 
    = - \sum_{a = 1}^n \sum_{k, l = 1}^N
        \{ \bsA^{(0)}_{k, l}, m_a \}^\ext (u^{(0)}) 
        (e_{a, a} + e_{n + a, n + a}) \otimes e_{k, l} \\
    & \quad
    = - \sum_{a = 1}^n 
            D^+_a 
            \otimes 
            \sqrt{2} 
                \left( 
                    \cA^{(0)} (\nabla^\mfg (\Imag(m_a)))(u^{(0)})
                    +(\nabla^\mfg (\Imag(m_a)))(u^{(0)}) \cA^{(0)} 
                \right).
\end{split}
\label{L7_PB_m_and_bsA_0}
\ee
Therefore, by projecting the Lie algebra element $\zeta_a$ 
(\ref{zeta_a_expanded}) onto the subspace $\mfm^\perp$, the application of 
(\ref{mfg_grad_imag_m_a}) immediately leads to (\ref{PB_m_and_bsA_0}). 
\end{proof}

\begin{LEMMA}
\label{LEMMA_PB_m_and_k}
At $u^{(0)}$ (\ref{u_0}) we have the trivial Poisson bracket 
$\{ k \stackrel{\otimes}{,} k \}^\ext (u^{(0)}) = 0$, whereas
\be
    \{ m \stackrel{\otimes}{,} k \}^\ext (u^{(0)}) 
    = \sum_{a = 1}^n D^+_a \otimes T^{(0)}_a,
\label{PB_m_and_k}
\ee
where for each $a \in \bN_n$ we have
\be
\begin{split}
    T^{(0)}_a
    = \frac{1}{\sqrt{2} \cF^{(0)}_a}
        & \left\{
            \frac{\sqrt{2} \Real(\cF^{(0)}_{n + a}) 
                    X^{+, \ri}_{2 \veps_a}}
                {2 \lambda^{(0)}_a}
        \right. \\
    & \quad 
    + \sum_{c = 1}^{a - 1} 
        \left(
            \frac{\cF^{(0)}_c 
                    X^{+, \ri}_{\veps_c - \veps_a}}
                {\lambda^{(0)}_c - \lambda^{(0)}_a} 
            + \frac{\Real(\cF^{(0)}_{n + c}) 
                    X^{+, \ri}_{\veps_c + \veps_a}}
                {\lambda^{(0)}_c + \lambda^{(0)}_a}
            - \frac{\Imag(\cF^{(0)}_{n + c}) 
                    X^{+, \rr}_{\veps_c + \veps_a}}
                {\lambda^{(0)}_c + \lambda^{(0)}_a}
        \right) \\
    & \left. \quad 
    + \sum_{c = a + 1}^n 
        \left(
            - \frac{\cF^{(0)}_c X^{+, \ri}_{\veps_a - \veps_c}}
                {\lambda^{(0)}_a - \lambda^{(0)}_c} 
            + \frac{\Real(\cF^{(0)}_{n + c}) 
                    X^{+, \ri}_{\veps_a + \veps_c}}
                {\lambda^{(0)}_a + \lambda^{(0)}_c}
            + \frac{\Imag(\cF^{(0)}_{n + c}) 
                    X^{+, \rr}_{\veps_a + \veps_c} }
                {\lambda^{(0)}_a + \lambda^{(0)}_c}
        \right)
    \right\}.
\end{split}
\label{T_a}
\ee
\end{LEMMA}

\begin{proof}
Working with the basis $\{ v_I \}$ (\ref{v_I}) of $\mfgl(N, \bC)$, from 
Proposition \ref{PROPOSITION_k_gradients} we see that
\be
    \nabla^\mfg (\Real(\tr(v_I k)))(u^{(0)}) 
        \in \mfa^\perp \subseteq \mfp.
\label{L8_bsA_gradient_mfg}
\ee
Keeping in mind the orthogonal $\bZ_2$-gradation (\ref{gradation}), the 
Poisson bracket formula (\ref{PB_ext}) gives rise to the relationship
\be
\begin{split}
    & \{ \tr(v_I k), \tr(v_J k) \}^\ext (u^{(0)}) \\
    & \quad
    = \left \langle
        [ \nabla^\mfg (\Real(\tr(v_I k)))(u^{(0)}),
        \nabla^\mfg (\Real(\tr(v_J k)))(u^{(0)}) ], \bsLambda^{(0)}
    \right \rangle
    = 0,
\end{split}
\label{L8_PB_k_and_k}
\ee
thus the equation $\{ k \stackrel{\otimes}{,} k \}^\ext (u^{(0)}) = 0$ 
follows immediately. 

As we have already observed at the beginning of the proof of Lemma 
\ref{LEMMA_PBs_with_bsA_0}, among the members of the basis $\{ v_I \}$ 
only the vectors $v = X^{+, \eps}_\alpha$ generate non-trivial
gradients for the component functions $\tr(v k)$ at the point $u^{(0)}$. 
Utilizing these gradients (\ref{L7_k_nontrivial_gradients}), Propositions 
\ref{PROPOSITION_k_gradients} and \ref{PROPOSITION_m_gradients} allow us 
to write
\be
\begin{split}
    & \{ m \stackrel{\otimes}{,} k \}^\ext(u^{(0)}) 
    = \sum_{a = 1}^n \sum_{\alpha, \eps}
        \{ m_a, \tr(X^{+, \eps}_\alpha k) \}^\ext(u^{(0)})
            (e_{a, a} + e_{n + a, n + a}) \otimes (-X^{+, \eps}_\alpha) \\
    & \quad
    = -\sum_{a = 1}^n \sum_{\alpha, \eps}
        \left \langle
            \nabla^G (\Imag(m_a))(u^{(0)}), 
            \nabla^\mfg (\Real(\tr(X^{+, \eps}_\alpha k)))(u^{(0)})
        \right \rangle
            \sqrt{2} D^+_a \otimes X^{+, \eps}_\alpha \\
    & \quad
    = - \sum_{a = 1}^n \frac{\sqrt{2}}{\cF^{(0)}_a} D^+_a 
        \otimes
        \sum_{\alpha, \eps}
            \langle 
                \zeta_a, 
                (\wad_{\bsLambda^{(0)}})^{-1} (X^{+, \eps}_\alpha) 
            \rangle 
            X^{+, \eps}_\alpha \\
    & \quad
    = - \sum_{a = 1}^n D^+_a 
        \otimes 
        \frac{\sqrt{2}}{\cF^{(0)}_a} 
            (\wad_{\bsLambda^{(0)}})^{-1}((\zeta_a)_{\mfa^\perp}).
\end{split}
\label{L8_PB_m_and_k}
\ee
Remembering the explicit formula of $\zeta_a$ (\ref{zeta_a_expanded})
and the commutation relations (\ref{commut_rel}), the Poisson bracket 
(\ref{PB_m_and_k}) also follows.
\end{proof}

\begin{LEMMA}
\label{LEMMA_PB_m_and_m}
At the distinguished point $u^{(0)}$ (\ref{u_0}) we have the Poisson bracket
\be
    \{ m \stackrel{\otimes}{,} m \}^\ext (u^{(0)}) 
    = \sum_{1 \leq a < b \leq n} \varPsi^{(0)}_{a, b}
        (D^+_a \otimes D^+_b - D^+_b \otimes D^+_a)
\label{PB_m_and_m}
\ee
with the coefficients
\be
    \varPsi^{(0)}_{a, b} 
    = \frac{1}{\lambda^{(0)}_a - \lambda^{(0)}_b}
        + \frac{\lambda^{(0)}_a - \lambda^{(0)}_b}
            {(\lambda^{(0)}_a - \lambda^{(0)}_b)^2 + 4 \mu^2}.
\label{varPsi}
\ee
\end{LEMMA}

\begin{proof}
Using the antisymmetry of the Poisson bracket, we find that
\be
\begin{split}
    & \{ m \stackrel{\otimes}{,} m \}^\ext(u^{(0)})
    = \sum_{a, b = 1}^n \{ m_a, m_b \}^\ext(u^{(0)})
        (e_{a, a} + e_{n + a, n + a}) 
        \otimes
        (e_{b, b} + e_{n + b, n + b}) \\
    & \quad
    = - 2 \sum_{1 \leq a < b \leq n} \{ m_a, m_b \}^\ext(u^{(0)})
        (D^+_a \otimes D^+_b - D^+_b \otimes D^+_a).
\end{split}
\label{L9_m_m_PB}
\ee
To proceed further, let us choose arbitrary $a, b \in \bN_n$ satisfying 
$a < b$. Notice that the Poisson bracket formula (\ref{PB_ext}) naturally 
leads to the expression
\be
\begin{split}
    \{ m_a, m_b \}^\ext(u^{(0)})
    = & - \langle 
        (\nabla^G (\Imag(m_a)))(u^{(0)}), 
        (\nabla^\mfg (\Imag(m_b)))(u^{(0)})
    \rangle \\
    & + \langle 
        (\nabla^\mfg (\Imag(m_a)))(u^{(0)}), 
        (\nabla^G (\Imag(m_b)))(u^{(0)})
    \rangle \\
    & - \omega^\cO_{\rho^{(0)}}
        ((\nabla^\cO (\Imag(m_a)))(u^{(0)}), 
        (\nabla^\cO (\Imag(m_b)))(u^{(0)})).
\end{split}
\label{L9_m_a_m_b_PB}
\ee
However, by utilizing Proposition \ref{PROPOSITION_m_gradients}, each term on 
the right hand side of the above equation can be cast into a fairly explicit 
form. Starting with the first term, the application of (\ref{zeta_a_expanded}) 
gives rise to the relationship
\be
\begin{split}
    & \langle 
        (\nabla^G (\Imag(m_a)))(u^{(0)}), 
        (\nabla^\mfg (\Imag(m_b)))(u^{(0)})
    \rangle 
    = \frac{1}{\cF^{(0)}_a \cF^{(0)}_b}
    \langle 
        \zeta_a, 
        (\wad_{\bsLambda^{(0)}})^{-1}((\zeta_b)_{\mfm^\perp})
    \rangle \\
    & \quad
    = \frac{1}{4 \cF^{(0)}_a \cF^{(0)}_b}
    \left(
        \frac{\cF^{(0)}_a \cF^{(0)}_b}
            {\lambda^{(0)}_a - \lambda^{(0)}_b}
        + \frac{\Real(\cF^{(0)}_{n + a}) \Real(\cF^{(0)}_{n + b})}
            {\lambda^{(0)}_a + \lambda^{(0)}_b}
        - \frac{\Imag(\cF^{(0)}_{n + a}) \Imag(\cF^{(0)}_{n + b})}
            {\lambda^{(0)}_a + \lambda^{(0)}_b}
    \right).
\end{split}
\label{L9_first_term}
\ee
Keeping in mind that $a < b$, a similar argument provides
\be
\begin{split}
    & \langle 
        (\nabla^\mfg (\Imag(m_a)))(u^{(0)}), 
        (\nabla^G (\Imag(m_b)))(u^{(0)})
    \rangle \\
    & \quad
    = \frac{1}{4 \cF^{(0)}_a \cF^{(0)}_b}
    \left(
        - \frac{\cF^{(0)}_a \cF^{(0)}_b}
            {\lambda^{(0)}_a - \lambda^{(0)}_b}
        + \frac{\Real(\cF^{(0)}_{n + a}) \Real(\cF^{(0)}_{n + b})}
            {\lambda^{(0)}_a + \lambda^{(0)}_b}
        - \frac{\Imag(\cF^{(0)}_{n + a}) \Imag(\cF^{(0)}_{n + b})}
            {\lambda^{(0)}_a + \lambda^{(0)}_b}
    \right).
\end{split}
\label{L9_second_term}
\ee
Now, let us turn to the third appearing in (\ref{L9_m_a_m_b_PB}). Utilizing 
the concise formula (\ref{omega_cO_OK}) for the symplectic form 
$\omega^\cO$ (\ref{omega_cO}), the application of the equations 
(\ref{cO_grad_imag_m_a}), (\ref{delta_V_a}) and (\ref{cA}) yields that
\be
\begin{split}
    & \omega^\cO_{\rho^{(0)}}
        ((\nabla^\cO (\Imag(m_a)))(u^{(0)}), 
        (\nabla^\cO (\Imag(m_b)))(u^{(0)})) 
    = \omega^\cO_{\xi(\cV^{(0)})}
        (\xi_{* \cV^{(0)}} (\delta V_a), \xi_{* \cV^{(0)}} (\delta V_b)) \\
    & \quad
    = 2 \mu \Imag((\delta V_a)^* \delta V_b)
    = 2 \mu \Imag 
    \left(
        \frac{\cA^{(0)}_{a, b}}{8 \mu^2 \cF^{(0)}_a \cF^{(0)}_b}
    \right) 
    = \half \frac{\lambda^{(0)}_a - \lambda^{(0)}_b}
                {(\lambda^{(0)}_a - \lambda^{(0)}_b)^2 + 4 \mu^2}.
\end{split}
\label{L9_third_term}
\ee
Now, by simply putting together the above equations, the Lemma follows 
at once.
\end{proof}

At this point we are in a position to provide an explicit formula for the 
$r$-matrix (\ref{r}). Remembering (\ref{varphi_bsA_0_PB}), let us notice
that Lemma \ref{LEMMA_PBs_with_bsA_0} itself implies that $r$ is in fact 
linear in $\cA$, having the form
\be
    r_{12}(\lambda^{(0)}, \theta^{(0)})
    = (p^+_{12})^{(0)} \cA^{(0)}_2 + \cA^{(0)}_2 (p^-_{12})^{(0)}
\label{r_linear_in_cA}
\ee
with the $\mfg \otimes \mfg$-valued functions
\be
\begin{split}
    (p^\pm_{12})^{(0)} 
    = & - \sum_{\alpha, \eps}
        \frac{X^{+, \eps}_\alpha \otimes X^{-, \eps}_\alpha}
            {\alpha(\lambda^{(0)})}
    + \sum_{a = 1}^n D^+_a \otimes S^{(0)}_a
    \pm \half \{ \varphi \stackrel{\otimes}{,} \varphi \}^\ext (u^{(0)}).
\end{split}
\label{p_pm}
\ee
Recalling (\ref{varphi_varphi_PB}), the above expressions can be further 
expanded. Indeed, by simply plugging the formulae displayed in Lemmas 
\ref{LEMMA_PB_m_and_k} and \ref{LEMMA_PB_m_and_m} into (\ref{p_pm}), we may
obtain explicit expressions for both $p^\pm_{12}$ and $r$. However, since 
$r$ is linear in $\cA$ as dictated by (\ref{r_linear_in_cA}), the linear 
$r$-matrix Poisson bracket (\ref{cA_PB}) can be cast into a \emph{quadratic} 
form. Also, since the point $(\lambda^{(0)}, \theta^{(0)})$ 
(\ref{fixed_point_in_cP_R}) we fixed at the beginning of Subsection
\ref{SUBSECTION_local_extension} was an arbitrary element of $\cP^R$, the
zero superscripts become superfluous and can be safely omitted. With the 
usual conventions for the symmetric and the antisymmetric tensor products,
\be
    X \vee Y = X \otimes Y + Y \otimes X
    \quad \text{and} \quad        
    X \wedge Y = X \otimes Y - Y \otimes X,
\label{symm_antisymm_convention}
\ee
we end up with the following result.

\begin{THEOREM}
\label{THEOREM_C_n_r_matrix}
The Lax matrix $\cA$ (\ref{cA}) of the rational $C_n$ RSvD model with
two independent coupling parameters obeys the quadratic Poisson bracket
\be
    \{ \cA \stackrel{\otimes}{,} \cA \}^R
    = \bsa_{12} \cA_1 \cA_2 + \cA_1 \bsb_{12} \cA_2 
        - \cA_2 \bsc_{12} \cA_1 - \cA_1 \cA_2 \bsd_{12}
\label{C_n_quadratic_PB}
\ee
with the $\mfg \otimes \mfg$-valued dynamical structure coefficients
\begin{align}
    & \bsa_{12} 
    = \sum_{\alpha, \eps} 
        \frac{X^{-, \eps}_\alpha \wedge X^{+, \eps}_\alpha}
            {\alpha(\lambda)}
    + \sum_{a = 1}^n D^+_a \wedge (S_a + T_a)
    + \sum_{1 \leq a < b \leq n} \varPsi_{a, b} D^+_a \wedge D^+_b, 
    \label{bsa_12} \\
    & \bsb_{12} 
    = \sum_{\alpha, \eps} 
        \frac{X^{-, \eps}_\alpha \vee X^{+, \eps}_\alpha}
            {\alpha(\lambda)}
    - \sum_{a = 1}^n (D^+_a \vee S_a + D^+_a \wedge T_a)
    - \sum_{1 \leq a < b \leq n} \varPsi_{a, b} D^+_a \wedge D^+_b, 
    \label{bsb_12} \\    
    & \bsc_{12} 
    = \sum_{\alpha, \eps} 
        \frac{X^{-, \eps}_\alpha \vee X^{+, \eps}_\alpha}
            {\alpha(\lambda)}
    - \sum_{a = 1}^n (D^+_a \vee S_a - D^+_a \wedge T_a)
    + \sum_{1 \leq a < b \leq n} \varPsi_{a, b} D^+_a \wedge D^+_b, 
    \label{bsc_12} \\
    & \bsd_{12} 
    = \sum_{\alpha, \eps} 
        \frac{X^{-, \eps}_\alpha \wedge X^{+, \eps}_\alpha}
            {\alpha(\lambda)}
    + \sum_{a = 1}^n D^+_a \wedge (S_a - T_a)
    - \sum_{1 \leq a < b \leq n} \varPsi_{a, b} D^+_a \wedge D^+_b,
    \label{bsd_12} 
\end{align}
where the constituent objects are defined in Lemmas 
\ref{LEMMA_PBs_with_bsA_0}, \ref{LEMMA_PB_m_and_k} and 
\ref{LEMMA_PB_m_and_m}.
\end{THEOREM}

\begin{proof}
Due to (\ref{r_linear_in_cA}), the Poisson bracket (\ref{cA_PB}) takes
the quadratic form (\ref{C_n_quadratic_PB}) with
\be
    \bsa_{12} = p^+_{12} - p^+_{21},
    \quad
    \bsb_{12} = - p^+_{12} - p^-_{21},
    \quad
    \bsc_{12} = - p^+_{21} - p^-_{12},
    \quad
    \bsd_{12} = p^-_{12} - p^-_{21}.
\label{T10_p_pm_and_bsabcd}
\ee
Remembering (\ref{p_pm}), (\ref{varphi_varphi_PB}), and the explicit formulae 
displayed in Lemmas \ref{LEMMA_PB_m_and_k} and \ref{LEMMA_PB_m_and_m}, the 
Theorem follows.
\end{proof}

We conclude this section with an important remark. Since the quadratic 
structure matrices (\ref{bsa_12}-\ref{bsd_12}) are derived from an $r$-matrix 
linear in $\cA$ as described in (\ref{r_linear_in_cA}), from the relationships 
(\ref{T10_p_pm_and_bsabcd}) it follows immediately that they satisfy the 
consistency conditions
\be
    \bsa_{21} = - \bsa_{12},
    \quad
    \bsd_{21} = - \bsd_{12},
    \quad
    \bsb_{21} = \bsc_{12},
    \quad
    \bsa_{12} + \bsb_{12} = \bsc_{12} + \bsd_{12}.
\label{consistency_conds}
\ee
The above observation can be paraphrased as follows. If a Lax matrix $\cA$ 
obeys a tensorial Poisson bracket (\ref{cA_PB}), and if the governing 
$r$-matrix is itself linear in $\cA$ as in (\ref{r_linear_in_cA}), then the
tensorial Poisson bracket can be rewritten as a quadratic bracket 
(\ref{C_n_quadratic_PB}) with quadratic structure matrices obeying the 
consistency conditions (\ref{consistency_conds}) automatically. It is a nice, 
but essentially trivial algebraic fact that the converse of this statement is 
also true. Indeed, suppose that a Lax matrix $\cA$ obeys a quadratic Poisson 
bracket (\ref{C_n_quadratic_PB}) with coefficients satisfying 
(\ref{consistency_conds}). Under these assumptions the quadratic bracket can 
be cast into a linear form (\ref{cA_PB}). More precisely, the governing 
$r$-matrix can be written in the form of (\ref{r_linear_in_cA}) with
\be
    p^+_{12} = \frac{\bsa_{12} + \bsu_{12}}{2}
    \quad \text{and} \quad
    p^-_{12} = \frac{\bsd_{12} - \bsb_{12} - \bsc_{12} - \bsu_{12}}{2},
\label{recovering_p_pm}
\ee
where $\bsu_{12}$ is an arbitrary $\mfg \vee \mfg$-valued function on the 
phase space, i.e. it obeys the symmetry condition $\bsu_{21} = \bsu_{12}$. 
This observation plays a crucial role in the developments of the next 
section.

\section{Classical $r$-matrix structure of the $BC_n$-type model}
\label{SECTION_BC_r_matrix}
\setcounter{equation}{0}
Utilizing a symplectic reduction framework, so far we have studied the
classical $r$-matrix structure for the rational $C_n$ RSvD model with two
independent coupling parameters $\mu$ and $\nu$. However, to handle the 
$BC_n$-type model as well, in this section we slightly change our point of 
view. Switching to a purely algebraic approach, we shall generalize Theorem 
\ref{THEOREM_C_n_r_matrix} to cover the most general rational $BC_n$ RSvD 
model with three independent coupling constants. As an added bonus, 
at the end of this section we will provide a Lax representation of the 
dynamics, too.

To describe the Lax matrix of the rational $BC_n$ RSvD system with the 
additional third real parameter $\kappa$, we need the functions
\be
    \bsalpha(x) = \frac{\sqrt{x + \sqrt{x^2 + \kappa^2}}}{\sqrt{2 x}}
    \quad \text{and} \quad
    \bsbeta(x) = \ri \kappa \frac{1}{\sqrt{2 x}}
                \frac{1}{\sqrt{x + \sqrt{x^2 + \kappa^2}}},
\label{bsalpha_and_bsbeta}
\ee
where $x \in (0, \infty)$. Also, with each 
$\lambda = (\lambda_1, \ldots, \lambda_n) \in \mfc$ we associate the 
group element
\be
    h(\lambda) 
        = \begin{bmatrix}
        \diag(\bsalpha(\lambda_1), \ldots, \bsalpha(\lambda_n)) 
            & \diag(\bsbeta(\lambda_1), \ldots, \bsbeta(\lambda_n)) 
        \\
        -\diag(\bsbeta(\lambda_1), \ldots, \bsbeta(\lambda_n)) 
            & \diag(\bsalpha(\lambda_1), \ldots, \bsalpha(\lambda_n))
        \end{bmatrix}
    \in G.
\label{h}
\ee
In \cite{Pusztai_NPB2012} we proved that the smooth function
$\tilde{\cA} \colon \cP^R \rightarrow G$ defined by the formula
\be
    \tcA (\lambda, \theta)
    = h(\lambda)^{-1} \cA(\lambda, \theta) h(\lambda)^{-1}
    \qquad
    ((\lambda, \theta) \in \cP^R)
\label{tcA}
\ee
provides a Lax matrix for the rational $BC_n$ RSvD model (\ref{H_R}) with 
the independent coupling parameters $\mu$, $\nu$ and $\kappa$. Our first 
goal in this section to construct a quadratic algebra relation for the 
Lax matrix $\tcA$ with structure coefficients satisfying the consistency 
conditions analogous to (\ref{consistency_conds}).

Recalling (\ref{D_-_c}) and (\ref{cA}), we start with the observation
\be
    \PD{\cA}{\theta_c}
    = \sqrt{2}(D^-_c \cA + \cA D^-_c)
    \qquad
    (c \in \bN_n).
\label{cA_theta_der}
\ee
Therefore, upon introducing the $\mfg \otimes \mfg$-valued function
\be
    \Gamma_{12} 
    = \frac{1}{\sqrt{2}} 
        \sum_{c = 1}^n D^-_c \otimes \PD{(h^{-1})}{\lambda_c},
\label{Gamma_12}
\ee
we can write the tensorial Poisson bracket 
\be
    \{ \cA_1, h^{-1}_2 \}^R = \Gamma_{12} \cA_1 + \cA_1 \Gamma_{12}.
\label{A_h_inv_PB}
\ee
Now, by simply applying the Leibniz rule, from (\ref{C_n_quadratic_PB}) we 
get that
\be
    \{ \tcA \stackrel{\otimes}{,} \tcA \}^R
    = \tbsa_{12} \tcA_1 \tcA_2 + \tcA_1 \tbsb_{12} \tcA_2 
        - \tcA_2 \tbsc_{12} \tcA_1 - \tcA_1 \tcA_2 \tbsd_{12}
\label{tcA_tcA_PB}
\ee
with the dynamical coefficients
\begin{align}
    & \tbsa_{12} 
    = h^{-1}_1 h^{-1}_2 \bsa_{12} h_1 h_2
        + h^{-1}_1 \Gamma_{12} h_1 h_2
        - h^{-1}_2 \Gamma_{21} h_1 h_2, 
    \label{tbsa_12} \\
    & \tbsb_{12}
    = h_1 h^{-1}_2 \bsb_{12} h^{-1}_1 h_2
        + h_1 \Gamma_{12} h^{-1}_1 h_2
        - h_1 h^{-1}_2 \Gamma_{21} h_2,
    \label{tbsb_12} \\
    & \tbsc_{12}
    = h^{-1}_1 h_2 \bsc_{12} h_1 h^{-1}_2
        - h^{-1}_1 h_2 \Gamma_{12} h_1
        + h_2 \Gamma_{21} h_1 h^{-1}_2,
    \label{tbsc_12} \\
    & \tbsd_{12}
    = h_1 h_2 \bsd_{12} h^{-1}_1 h^{-1}_2
        - h_1 h_2 \Gamma_{12} h^{-1}_1
        + h_1 h_2 \Gamma_{21} h^{-1}_2.
    \label{tbsd_12}
\end{align}
Since the decorations coming from $h$ are `equally distributed' among these
new functions, we expect that likewise they satisfy the consistency conditions 
analogous to (\ref{consistency_conds}). Somewhat surprisingly, this naive idea 
is fully confirmed by the following result.

\begin{THEOREM}
\label{THEOREM_BC_n_quadratic_PB}
The functions (\ref{tbsa_12}-\ref{tbsd_12}) appearing in the tensorial 
Poisson bracket (\ref{tcA_tcA_PB}) obey the consistency conditions
\be
    \tbsa_{21} = - \tbsa_{12},
    \quad
    \tbsd_{21} = - \tbsd_{12},
    \quad
    \tbsb_{21} = \tbsc_{12},
    \quad
    \tbsa_{12} + \tbsb_{12} = \tbsc_{12} + \tbsd_{12}.
\label{tconsistency_conds}
\ee
In other words, the Lax matrix $\tcA$ (\ref{tcA}) of the rational $BC_n$
RSvD system satisfies a quadratic Poisson bracket (\ref{tcA_tcA_PB}) 
characterized by the consistent dynamical structure coefficients 
(\ref{tbsa_12}-\ref{tbsd_12}). 
\end{THEOREM}

\begin{proof}
A moment of reflection reveals that $\tbsa_{21} = - \tbsa_{12}$,
$\tbsd_{21} = - \tbsd_{12}$, and $\tbsb_{21} = \tbsc_{12}$, whence it is 
enough to prove that $\tbsa_{12} + \tbsb_{12} = \tbsc_{12} + \tbsd_{12}$.
Since the verification of this last equation is basically an involved 
algebraic computation, in the following we wish to highlight only the key 
steps. First, we introduce the functions
\be
    \cP(x) = \sqrt{1 + \frac{\kappa^2}{x^2}} 
    \quad \text{and} \quad
    \cQ(x) = \frac{\ri \kappa}{x}
    \qquad
    (x \in (0, \infty)).
\label{T11_cP_and_cQ}
\ee
Remembering (\ref{bsalpha_and_bsbeta}), we see that 
\be
    \cP(x) = \bsalpha(x)^2 - \bsbeta(x)^2 
    \quad \text{and} \quad
    \cQ(x) = 2 \bsalpha(x) \bsbeta(x). 
\label{T11_cP_cQ_bsalpha_bsbeta}
\ee    
To make the presentation a slightly simpler, we also introduce the 
$G$-valued function 
\be
    \cH = h^2. 
\label{T11_cH_def}
\ee    
Now, recalling (\ref{bsalpha_and_bsbeta}), (\ref{h}) and 
(\ref{T11_cP_cQ_bsalpha_bsbeta}), with the notations
\be
    \cP_a = \cP(\lambda_a)
    \quad \text{and} \quad
    \cQ_a = \cQ(\lambda_a)
    \qquad
    (a \in \bN_n)
\label{T11_cP_a_and_cQ_a_def}
\ee
we can clearly write
\be
    \cH
    = \begin{bmatrix}
        \diag(\cP_1, \ldots, \cP_n) 
            & \diag(\cQ_1, \ldots, \cQ_n) 
        \\
        -\diag(\cQ_1, \ldots, \cQ_n) 
            & \diag(\cP_1, \ldots, \cP_n)
    \end{bmatrix}.
\label{T11_cH}
\ee
To proceed further, we also define the $\mfg \otimes \mfg$-valued function
\be
    \Omega_{12} = -(h^{-1}_2 \Gamma_{12} + \Gamma_{12} h^{-1}_2) \cH_2.
\label{T11_Omega_12}
\ee
Remembering the form of $\Gamma_{12}$ (\ref{Gamma_12}), Leibniz rule yields
\be
    \Omega_{12} 
    = \frac{1}{\sqrt{2}} 
        \sum_{c = 1}^n 
            D^-_c 
            \otimes 
            \left( \cH^{-1} \PD{\cH}{\lambda_c} \right),
\label{T11_Omega_12_OK}
\ee
where for the derivatives we can easily find that
\be
    \cH^{-1} \PD{\cH}{\lambda_c}
    = \frac{\sqrt{2} \kappa}
        {\lambda_c \sqrt{\lambda_c^2 + \kappa^2}} X^{-, \ri}_{2 \veps_c}
    \qquad
    (c \in \bN_n).
\label{T11_cH_der}
\ee
Bearing in mind the above objects, from (\ref{tbsa_12}-\ref{tbsd_12}) 
one can derive that
\be
\begin{split}
    & h^{-1}_1 h^{-1}_2 
        (\tbsa_{12} + \tbsb_{12} - \tbsc_{12} - \tbsd_{12}) h_1 h_2 \\
    & \quad
    = \cH^{-1}_1 \cH^{-1}_2 \bsa_{12} \cH_1 \cH_2
        + \cH^{-1}_2 \bsb_{12} \cH_2
        - \cH^{-1}_1 \bsc_{12} \cH_1
        - \bsd_{12} \\
    & \qquad
        - \cH^{-1}_1 \Omega_{12} \cH_1
        + \cH^{-1}_2 \Omega_{21} \cH_2
        - \Omega_{12}
        + \Omega_{21}.
\end{split}
\label{T11_tbsabcd_and_cH}
\ee
To handle the right hand side of the above equation, we need the commutation
relations listed below. First, for each $c \in \bN_n$ we have
\begin{align}
    & \cH^{-1} D^+_c \cH = D^+_c,
    \label{CR_D_+} \\
    & \cH^{-1} D^-_c \cH 
    = (\cP^2_c - \cQ^2_c) D^-_c + 2 \ri \cP_c \cQ_c X^{+, \ri}_{2 \veps_c},
    \label{CR_D_-} \\
    & \cH^{-1} X^{+, \ri}_{2 \veps_c} \cH 
    = (\cP^2_c - \cQ^2_c) X^{+, \ri}_{2 \veps_c} 
        + 2 \ri \cP_c \cQ_c D^-_c,
    \label{CR_Xpi2} \\
    & \cH^{-1} X^{-, \ri}_{2 \veps_c} \cH = X^{-, \ri}_{2 \veps_c}.
    \label{CR_Xmi2} 
\end{align} 
Also, if $a, b \in \bN_n$ and $a < b$, then we can write
\begin{align}
    & \cH^{-1} X^{+, \rr}_{\veps_a - \veps_b} \cH
    = (\cP_a \cP_b + \cQ_a \cQ_b) X^{+, \rr}_{\veps_a - \veps_b}
        + \ri (\cP_a \cQ_b - \cP_b \cQ_a) X^{-, \ri}_{\veps_a + \veps_b},
    \label{CR_XprM} \\
    & \cH^{-1} X^{-, \rr}_{\veps_a - \veps_b} \cH
    = (\cP_a \cP_b - \cQ_a \cQ_b) X^{-, \rr}_{\veps_a - \veps_b}
        + \ri (\cP_a \cQ_b + \cP_b \cQ_a) X^{+, \ri}_{\veps_a + \veps_b},
    \label{CR_XmrM} \\
    & \cH^{-1} X^{+, \rr}_{\veps_a + \veps_b} \cH
    = (\cP_a \cP_b - \cQ_a \cQ_b) X^{+, \rr}_{\veps_a + \veps_b}
        - \ri (\cP_a \cQ_b + \cP_b \cQ_a) X^{-, \ri}_{\veps_a - \veps_b},
    \label{CR_XprP} \\
    & \cH^{-1} X^{-, \rr}_{\veps_a + \veps_b} \cH
    = (\cP_a \cP_b + \cQ_a \cQ_b) X^{-, \rr}_{\veps_a + \veps_b}
        - \ri (\cP_a \cQ_b - \cP_b \cQ_a) X^{+, \ri}_{\veps_a - \veps_b},
    \label{CR_XmrP}
\end{align}
together with the relations
\begin{align}
    & \cH^{-1} X^{+, \ri}_{\veps_a - \veps_b} \cH
    = (\cP_a \cP_b + \cQ_a \cQ_b) X^{+, \ri}_{\veps_a - \veps_b}
        - \ri (\cP_a \cQ_b - \cP_b \cQ_a) X^{-, \rr}_{\veps_a + \veps_b},
    \label{CR_XpiM} \\
    & \cH^{-1} X^{-, \ri}_{\veps_a - \veps_b} \cH
    = (\cP_a \cP_b - \cQ_a \cQ_b) X^{-, \ri}_{\veps_a - \veps_b}
        - \ri (\cP_a \cQ_b + \cP_b \cQ_a) X^{+, \rr}_{\veps_a + \veps_b},
    \label{CR_XmiM} \\
    & \cH^{-1} X^{+, \ri}_{\veps_a + \veps_b} \cH
    = (\cP_a \cP_b - \cQ_a \cQ_b) X^{+, \ri}_{\veps_a + \veps_b}
        + \ri (\cP_a \cQ_b + \cP_b \cQ_a) X^{-, \rr}_{\veps_a - \veps_b},
    \label{CR_XpiP} \\
    & \cH^{-1} X^{-, \ri}_{\veps_a + \veps_b} \cH
    = (\cP_a \cP_b + \cQ_a \cQ_b) X^{-, \ri}_{\veps_a + \veps_b}
        + \ri (\cP_a \cQ_b - \cP_b \cQ_a) X^{+, \rr}_{\veps_a - \veps_b}.
    \label{CR_XmiP}
\end{align}
Now, let us examine the first four terms appearing on the right hand side 
of (\ref{T11_tbsabcd_and_cH}). Recalling (\ref{bsa_12}-\ref{bsd_12}), the 
application of (\ref{CR_D_+}) itself yields the formula
\be
\begin{split}
    & \cH^{-1}_1 \cH^{-1}_2 \bsa_{12} \cH_1 \cH_2
        + \cH^{-1}_2 \bsb_{12} \cH_2
        - \cH^{-1}_1 \bsc_{12} \cH_1
        - \bsd_{12} \\
    & \quad
    = \sum_{\alpha, \eps}
        \frac{(\cH^{-1} X^{-, \eps}_\alpha \cH + X^{-, \eps}_\alpha) 
            \wedge
            (\cH^{-1} X^{+, \eps}_\alpha \cH - X^{+, \eps}_\alpha)}
        {\alpha(\lambda)}.
\end{split}
\label{T11_part_1_step_1}
\ee
However, in order to further simplify this expression, we still have to 
exploit some functional equations obeyed by $\cP$ and $\cQ$. By inspecting 
the definitions (\ref{T11_cP_and_cQ}), we see immediately that
\be
    \cP(x)^2 - \cQ(x)^2 = 1 + \frac{2 \kappa^2}{x^2}
    \qquad
    (x \in (0, \infty)).
\label{T11_funct_eq_1}
\ee
A slightly longer calculation also reveals that
\begin{align}
    & \frac{\cP(x)^2 \cP(y)^2 - (\cQ(x) \cQ(y) - 1)^2}{x - y}
    + \frac{\cP(x)^2 \cQ(y)^2 - \cP(y)^2 \cQ(x)^2}{x + y}
    = 0,
    \label{T11_funct_eq_2} \\
    & \frac{\cP(x)^2 \cQ(y)^2 - \cP(y)^2 \cQ(x)^2}{x - y}
    + \frac{\cP(x)^2 \cP(y)^2 - (\cQ(x) \cQ(y) + 1)^2}{x + y}
    = 0,
    \label{T11_funct_eq_3}
\end{align}
where $x, y \in (0, \infty)$ and $x \neq y$. Having equipped with the 
relations (\ref{T11_funct_eq_1}-\ref{T11_funct_eq_3}), let us note that the
application of the commutation relations (\ref{CR_D_+}-\ref{CR_XmiP}) does 
give rise to an even greater simplification in (\ref{T11_part_1_step_1}). 
Indeed, we find that
\be
\begin{split}
    & \sum_{\alpha, \eps}
        \frac{(\cH^{-1} X^{-, \eps}_\alpha \cH + X^{-, \eps}_\alpha) 
            \wedge
            (\cH^{-1} X^{+, \eps}_\alpha \cH - X^{+, \eps}_\alpha)}
        {\alpha(\lambda)} \\
    & \quad
    = \sum_{c = 1}^n 
        \left(
            \frac{2 \kappa^2}{\lambda^3_c} 
            X^{-, \ri}_{2 \veps_c} \wedge X^{+, \ri}_{2 \veps_c}
            + \frac{2 \kappa}{\lambda^2_c} 
                \sqrt{1 + \frac{\kappa^2}{\lambda^2_c}}
                D^-_c \wedge X^{-, \ri}_{2 \veps_c}
        \right).   
\end{split}
\label{T11_part_1_step_2}
\ee
Now, let us turn to the last four terms appearing on the right hand side of 
(\ref{T11_tbsabcd_and_cH}). Recalling (\ref{T11_Omega_12_OK}), 
(\ref{T11_cH_der}), and (\ref{T11_funct_eq_1}), we can write that
\be
    \cH^{-1}_1 \Omega_{12} \cH_1 + \Omega_{12}
    = \sum_{c = 1}^n
        \left(
            - \frac{2 \kappa^2}{\lambda_c^3}
                X^{+, \ri}_{2 \veps_c} \otimes X^{-, \ri}_{2 \veps_c}
            + \frac{2 \kappa}{\lambda^2_c} 
                \sqrt{1 + \frac{\kappa^2}{\lambda^2_c}}
                D^-_c \otimes X^{-, \ri}_{2 \veps_c}
        \right).
\label{T11_part_2_key}
\ee
Now, by plugging (\ref{T11_part_1_step_1}), (\ref{T11_part_1_step_2}) and 
(\ref{T11_part_2_key}) into (\ref{T11_tbsabcd_and_cH}), we obtain at once that
\be
    h^{-1}_1 h^{-1}_2 
        (\tbsa_{12} + \tbsb_{12} - \tbsc_{12} - \tbsd_{12}) 
        h_1 h_2
    = 0,
\label{T11_done}
\ee
whence the proof is complete.
\end{proof}

Having completed the proof, now we offer a few remarks on the result. 
First, since the Lax matrix $\tcA$ obeys the quadratic bracket 
(\ref{tcA_tcA_PB}) with the dynamical objects (\ref{tbsa_12}-\ref{tbsd_12}) 
satisfying the consistency conditions (\ref{tconsistency_conds}), the 
quadratic bracket (\ref{tcA_tcA_PB}) can be rewritten as
\be
    \{ \tcA_1, \tcA_2 \}^R
    = [\tilde{r}_{12}, \tcA_1] - [\tilde{r}_{21}, \tcA_2].
\label{tilde_r_bracket}
\ee
Indeed, recalling our discussion at the end of the previous section, an
appropriate $r$-matrix is provided by the formula
\be
    \tilde{r}_{12} = \tilde{p}^+_{12} \tcA_2 + \tcA_2 \tilde{p}^-_{12},
\label{tilde_r}
\ee
where
\be
    \tilde{p}^+_{12} = \frac{\tbsa_{12} + \tbsu_{12}}{2}
    \quad \text{and} \quad
    \tilde{p}^-_{12} 
        = \frac{\tbsd_{12} - \tbsb_{12} - \tbsc_{12} - \tbsu_{12}}{2},
\label{tilde_p_pm}
\ee
with an arbitrary $\mfg \vee \mfg$-valued dynamical object $\tbsu_{12}$.

Second, one may raise the objection that the formulae 
(\ref{tbsa_12}-\ref{tbsd_12}) for the quadratic structure matrices in the 
$BC_n$ case are `less explicit' than the analogous objects 
(\ref{bsa_12}-\ref{bsd_12}) in the $C_n$ case. The trouble is mainly caused 
by the derivatives of $h^{-1}$ appearing in the definition of $\Gamma_{12}$ 
(\ref{Gamma_12}). Though these derivatives can be worked out rather easily, 
we propose an alternative approach to cure the problem. Namely, let us apply 
the gauge transformation
\be
    \hcA(\lambda, \theta) 
    = h(\lambda) \tcA(\lambda, \theta) h(\lambda)^{-1}
    \qquad
    ((\lambda, \theta) \in \cP^R). 
\label{hcA}
\ee
on the Lax matrix $\tcA$ (\ref{tcA}). By applying the corresponding
transformation on $\tilde{r}$ (\ref{tilde_r}), it turns out that the 
transformed $r$-matrix takes the form
\be
    \hat{r}_{12} = \hat{p}^+_{12} \hcA_2 + \hcA_2 \hat{p}^-_{12},
\label{hat_r}
\ee
where
\be
    \hat{p}^+_{12} 
    = h_1 h_2 \tilde{p}^+_{12} h^{-1}_1 h^{-1}_2 + h_1 \Gamma_{21}
    \quad \text{and} \quad
    \hat{p}^-_{12} 
    = h_1 h_2 \tilde{p}^-_{12} h^{-1}_1 h^{-1}_2 
        + h_1 \cH_2 \Gamma_{21} \cH^{-1}_2.
\label{hat_p_pm}
\ee
Since $\hat{r}_{12}$ is linear in $\hcA$, the tensorial Poisson bracket for 
$\hcA$ can be cast into a quadratic form with structure matrices obeying 
the consistency conditions analogous to (\ref{tconsistency_conds}). To save 
time on the algebraic details, we present only the resulting formulae.

\begin{THEOREM}
\label{THEOREM_BC_n_transformed_r_matrices}
For the transformed Lax matrix $\hcA = \cA \cH^{-1}$ of the rational $BC_n$ 
RSvD model we have
\be
    \{ \hcA_1, \hcA_2 \}^R
    = \hbsa_{12} \hcA_1 \hcA_2 + \hcA_1 \hbsb_{12} \hcA_2 
        - \hcA_2 \hbsc_{12} \hcA_1 - \hcA_1 \hcA_2 \hbsd_{12},
\label{hcA_hcA_PB}
\ee
where the dynamical objects
\begin{align}
    & \hbsa_{12} = \bsa_{12}, 
    \label{hbsa_12} \\
    & \hbsb_{12} 
    = \cH_1 
    \left(
        \bsb_{12} 
        + \frac{1}{\sqrt{2}} 
            \sum_{c = 1}^n 
            \left( \cH^{-1} \PD{\cH}{\lambda_c} \right)
            \otimes
            D^-_c
    \right)
    \cH^{-1}_1,
    \label{hbsb_12} \\
    & \hbsc_{12} 
    = \cH_2 
    \left(
        \bsc_{12} 
        + \frac{1}{\sqrt{2}} 
            \sum_{c = 1}^n 
            D^-_c
            \otimes
            \left( \cH^{-1} \PD{\cH}{\lambda_c} \right)
        \right)
    \cH^{-1}_2,
    \label{hbsc_12} \\
    & \hbsd_{12} 
    = \cH_1 \cH_2 
    \left(
        \bsd_{12} 
        + \frac{1}{\sqrt{2}} 
            \sum_{c = 1}^n
            D^-_c
            \wedge 
            \left( \cH^{-1} \PD{\cH}{\lambda_c} \right)
    \right)
    \cH^{-1}_1 \cH^{-1}_2
    \label{hbsd_12}
\end{align}
are built up from the explicitly given functions (\ref{bsa_12}-\ref{bsd_12}), 
(\ref{T11_cH}) and (\ref{T11_cH_der}). Furthermore, by construction, they 
satisfy the consistency conditions
\be
    \hbsa_{21} = - \hbsa_{12},
    \quad
    \hbsd_{21} = - \hbsd_{12},
    \quad
    \hbsb_{21} = \hbsc_{12},
    \quad
    \hbsa_{12} + \hbsb_{12} = \hbsc_{12} + \hbsd_{12}.
\label{hconsistency_conds}
\ee
\end{THEOREM}

As an immediate consequence of the $r$-matrix formalism, now we can easily
construct a Lax pair for the rational $BC_n$ RSvD system. For this reason, 
let us recall the partial trace operation on the second factor of 
$\mfgl(N, \bC) \otimes \mfgl(N, \bC)$, which is uniquely determined by the 
condition
\be
    \tr_2 (X \otimes Y) = \tr(Y) X 
    \qquad
    (X, Y \in \mfgl(N, \bC)).
\label{tr_2}
\ee
As we proved in \cite{Pusztai_NPB2012}, for the Hamiltonian of the rational
$BC_n$ RSvD model we can write
\be
    H^R = \half \tr (\tcA) = \half \tr (\hcA).
\label{H_R_and_hcA}
\ee
Therefore, by expanding the Lax matrix $\hcA$ (\ref{hcA}) in an arbitrary 
basis of the complex linear space $\mfgl(N, \bC)$, say in the basis 
$\{ v_I \}$ (\ref{v_I}), we can write
\be
    H^R 
    = \half \tr \left( \sum_J \hcA^J v_J \right) 
    = \half \sum_J \hcA^J \tr(v_J). 
\label{H_R_and_expansion}
\ee
Making use of this expansion, for the action of the Hamiltonian vector field 
$\bsX_{H^R} \in \mfX(\cP^R)$ on the Lax operator $\hcA$ we obtain
\be
\begin{split}
    & \bsX_{H^R}[\hcA] 
    = \sum_I \bsX_{H^R}[\hcA^I] v_I 
    = \sum_I \{ \hcA^I, H^R \}^R v_I \\
    & \quad
    = \half \sum_{I, J} \{ \hcA^I, \hcA^J \}^R \tr(v_J) v_I
    = \half \tr_2 
        \left( 
            \sum_{I, J} \{ \hcA^I, \hcA^J \}^R v_I \otimes v_J 
        \right) \\
    & \quad = \half \tr_2 
        \left(
            [\hat{r}_{12}, \hcA_1] - [\hat{r}_{21}, \hcA_2]
        \right)
    = \half [ \tr_2(\hat{r}_{12}), \hcA ].
\end{split}
\label{mfX_H_R_on_hcA}
\ee
Since for $\hat{p}^\pm_{12}$ (\ref{hat_p_pm}) we can write
\be
    \hat{p}^+_{12} = \frac{\hbsa_{12} + \hbsu_{12}}{2}
    \quad \text{and} \quad
    \hat{p}^-_{12} 
        = \frac{\hbsd_{12} - \hbsb_{12} - \hbsc_{12} - \hbsu_{12}}{2},
\label{hat_p_pm_and_hbsabcdu}
\ee
where $\hbsu_{12}$ is a $\mfg \vee \mfg$-valued function, the application 
of (\ref{hat_r}) and (\ref{hconsistency_conds}) yields that
\be
    \tr_2(\hat{r}_{12}) = \tr_2((\hbsa_{12} - \hbsc_{12}) \hcA_2).
\label{tr_2_hat_r_12}
\ee
Upon introducing the alternative Lax matrix
\be
    \ccA = \cH^{-1} \hcA \cH = \cH^{-1} \cA,
\label{ccA}
\ee 
the combination of the formulae (\ref{hbsa_12}), (\ref{hbsc_12}), 
(\ref{mfX_H_R_on_hcA}) and (\ref{tr_2_hat_r_12}), together with 
the explicit expressions (\ref{bsa_12}) and (\ref{bsc_12}), leads 
to the following result.

\begin{THEOREM}
\label{THEOREM_Lax_pair}
With the aid of the $\mfg$-valued function
\be
\begin{split}
    \hcB 
    & = \half \sum_{\alpha, \eps} 
        \frac{\tr(X^{+, \eps}_\alpha (\hcA - \ccA)) X^{-, \eps}_\alpha
            - \tr(X^{-, \eps}_\alpha (\hcA + \ccA)) X^{+, \eps}_\alpha}
        {\alpha(\lambda)} \\
    & \quad
        + \half \sum_{c = 1}^n 
            \tr(S_c (\hcA + \ccA) + T_c (\hcA - \ccA) ) D^+_c
        - \frac{\kappa}{2} \sum_{c = 1}^n 
            \frac{\tr(X^{-, \ri}_{2 \veps_c} \ccA)}
            {\lambda_c \sqrt{\lambda^2_c + \kappa^2}} D^-_c
\end{split}
\label{cB}
\ee
the derivative of the Lax matrix $\hcA$ (\ref{hcA}) along the Hamiltonian
vector field $\bsX_{H^R}$ takes the Lax form 
$\bsX_{H^R} [\hcA] = [\hcB, \hcA]$. In other words, $\hcB$ provides a Lax 
pair for $\hcA$.
\end{THEOREM}

\section{Discussion}
\label{SECTION_Discussion}
\setcounter{equation}{0}
One of the most important objects in the algebraic formulation of the 
theory of classical integrable systems is undoubtedly the $r$-matrix 
structure encoding the tensorial Poisson bracket of the Lax matrix. In 
the context of the $A_n$-type CMS and RSvD models the underlying dynamical 
$r$-matrix structure is under complete control, even in the elliptic case 
(see e.g. \cite{Avan_Talon_1993, Sklyanin_1994, Braden_Suzuki_1994, 
Nijhoff_et_al_1996}). In sharp contrast, for the models associated with 
the non-$A_n$-type root systems the theory is far less developed. By 
generalizing the ideas of Avan, Babelon and Talon 
\cite{Avan_Babelon_Talon_AA1994}, in our earlier paper \cite{Pusztai_JMP2012}
we constructed a dynamical $r$-matrix structure for the most general 
hyperbolic $BC_n$ Sutherland system with three independent coupling 
constants. However, for the elliptic case only partial results are available
\cite{Forger_Winterhalder_2002}. For the non-$A_n$-type RSvD systems
the situation is even more delicate. Prior to our present paper, the
$r$-matrix structure of the $BC_n$ RSvD systems was studied only in 
\cite{Avan_Rollet_2002}, based on the special one-parameter family of Lax 
matrices coming from $\bZ_2$-folding of the $A_{2 n}$ root system. 
Nevertheless, in the present paper we succeeded in constructing a 
quadratic $r$-matrix structure for the rational $BC_n$ RSvD systems with 
the maximal number of three coupling parameters, as formulated in Theorems 
\ref{THEOREM_BC_n_quadratic_PB} and \ref{THEOREM_BC_n_transformed_r_matrices}. 
It is also clear that by applying a standard analytic continuation argument 
on our formulae, one can easily derive a dynamical $r$-matrix structure for 
the rational RSvD system appearing in \cite{Feher_Gorbe_2014}. 

Regarding the hyperbolic, trigonometric and elliptic variants of the 
non-$A_n$-type RSvD systems we also face many interesting questions. Indeed, 
except from some very special cases \cite{Chen_et_al_JMP2000, 
Chen_et_al_JMP2001, Chen_Hou_2001}, even the construction of Lax matrices for
these models is a wide open problem. However, let us note that in the last 
couple of years many results for the $A_n$-type models have been reinterpreted 
in a more geometric context using advanced techniques from the theory of 
reductions (see e.g. \cite{Feher_Klimcik_CMP2011, Feher_Klimcik_NPB2012}). 
Relatedly, it would be of considerable interest to see whether the underlying
classical $r$-matrix structures can be explored from these geometric pictures 
along the line of our present paper. We also expect that the various reduction 
approaches eventually may lead to a progress in the rigorous geometric theory 
of the non-$A_n$-type trigonometric, hyperbolic and elliptic RSvD systems as 
well. As a starting point, it is worth mentioning the recent paper 
\cite{Marshall_CMP2015}, in which a Hamiltonian reduction approach based on 
the Heisenberg double of $SU(n, n)$ gives rise to a new integrable particle 
system, that in the cotangent bundle limit gives back the familiar hyperbolic 
$BC_n$ Sutherland model with three independent coupling parameters.

Turning back to our quadratic $r$-matrix algebra (\ref{hcA_hcA_PB}), let us 
observe that the structure matrices $\hbsa_{12}$, $\hbsb_{12}$, $\hbsc_{12}$ 
and $\hbsd_{12}$ are fully dynamical, i.e. they depend on all variables of 
the phase space $\cP^R$ (\ref{cP_R}) in an essential way. It is in contrast 
with the CMS models, where the naturally appearing dynamical $r$-matrices 
usually depend only on the configuration space variables. Moreover, in many
variants of the CMS models the $r$-matrices can be related to the dynamical 
Yang--Baxter equation, as first realized in \cite{Avan_Babelon_Billey_1996}. 
However, in the $A_n$ case Suris \cite{Suris_1997} observed that in some 
special choice of gauge the CMS and the RSvD models can be characterized by 
the same dynamical $r$-matrices. Working in this gauge, Nagy, Avan and Rollet 
proved that the quadratic structure matrices of the hyperbolic $A_n$ RSvD 
system do obey certain dynamical quadratic Yang--Baxter equations (see 
Proposition 1 in \cite{Nagy_Avan_Rollet}, and relatedly also 
\cite{Arutyunov_Chekov_Frolov}). As a natural next step, we find it an 
important question whether such claims can be made about the quadratic 
algebra relation (\ref{hcA_hcA_PB}) in an appropriate gauge. Also, it would 
be of considerable interest to investigate whether the non-$A_n$-type RSvD 
models can be characterized by numerical, i.e. non-dynamical $r$-matrices. 
In the $A_n$ case the answer is in the affirmative (see \cite{Hou_Wang_2000}), 
but in the $BC_n$ case the analogous tasks seem to be quite challenging even 
for the rational models. Nevertheless, we wish to come back to these problems 
in later publications.

\medskip
\noindent
\textbf{Acknowledgments.}
Our work was supported by the J\'anos Bolyai Research Scholarship of the 
Hungarian Academy of Sciences. This work was also supported by a Lend\"ulet 
Grant; we wish to thank Z.~Bajnok for hospitality in the MTA Lend\"ulet 
Holographic QFT Group.


\end{document}